\documentclass[letterpaper, 11pt, thm-restate]{article}

\usepackage[margin=1in]{geometry}
\usepackage{footmisc}
\usepackage{enumitem}
\usepackage[title]{appendix}
\usepackage{xcolor}
\usepackage{array}
\usepackage{extarrows}
\usepackage{amsthm}
\usepackage{amsmath}
\usepackage{amssymb}
\usepackage{amsfonts}
\usepackage{graphicx}
\usepackage[ruled]{algorithm}
\usepackage{thm-restate}
\usepackage[colorlinks = true]{hyperref}
\hypersetup{
    linkcolor=black,
    citecolor=blue,
    urlcolor=blue}

\usepackage{mathtools}

\usepackage{titlesec}
\titleformat{\subsection}[runin]
       {\normalfont\bfseries}
       {\thesubsection}
       {0.5em}
       {}
       [.]

\newtheorem{theorem}{Theorem}[section]

\newtheorem{proposition}[theorem]{Proposition}
\newtheorem{lemma}[theorem]{Lemma}

\newtheorem{observation}[theorem]{Observation}
\theoremstyle{definition}
\newtheorem{definition}[theorem]{Definition}

\newtheorem{example}[theorem]{Example}

\newcommand{\Z}{\mathbb{Z}}
\newcommand{\N}{\mathbb{N}}
\newcommand{\Q}{\mathbb{Q}}
\newcommand{\R}{\mathbb{R}}

\newcommand{\lcm}{\operatorname{lcm}}
\newcommand{\mG}{\mathcal{G}}

\newcommand{\mY}{\mathcal{Y}}
\newcommand{\mI}{\mathcal{I}}

\newcommand{\mJ}{\mathcal{J}}

\newcommand{\mH}{\mathcal{H}}
\newcommand{\mA}{\mathcal{A}}
\newcommand{\mM}{\mathcal{M}}

\newcommand{\mX}{\mathcal{X}}

\newcommand{\ba}{\boldsymbol{a}}
\newcommand{\bb}{\boldsymbol{b}}
\newcommand{\bc}{\boldsymbol{c}}
\newcommand{\bm}{\boldsymbol{m}}

\newcommand{\bzer}{\boldsymbol{0}}

\newcommand{\gen}[1]{\langle {#1} \rangle}
\newcommand{\bff}{\boldsymbol{f}}
\newcommand{\bg}{\boldsymbol{g}}

\newcounter{DecideCounter}

\begin{document}

\title{Linear equations with monomial constraints and decision problems in abelian-by-cyclic groups}
\author{Ruiwen Dong\footnote{Department of Mathematics, Saarland University. Email: ruiwen.dong@uni-saarland.de}}

\date{}
\maketitle
\thispagestyle{empty}

\begin{abstract}
    We show that it is undecidable whether a system of linear equations over the Laurent polynomial ring $\mathbb{Z}[X^{\pm}]$ admit solutions where a specified subset of variables take value in the set of monomials $\{X^z \mid z \in \mathbb{Z}\}$.
    In particular, we construct a finitely presented $\mathbb{Z}[X^{\pm}]$-module, where it is undecidable whether a linear equation $X^{z_1} \boldsymbol{f}_1 + \cdots + X^{z_n} \boldsymbol{f}_n = \boldsymbol{f}_0$ has solutions $z_1, \ldots, z_n \in \mathbb{Z}$.
    This contrasts the decidability of the case $n = 1$, which can be deduced from Noskov's Lemma.

    We apply this result to settle a number of problems in computational group theory. 
    We show that it is undecidable whether a system of equations has solutions in the wreath product $\mathbb{Z} \wr \mathbb{Z}$, providing a negative answer to an open problem of Kharlampovich, L\'{o}pez and Miasnikov (2020).
    We show that there exists a finitely generated abelian-by-cyclic group in which the problem of solving a single (spherical) quadratic equation is undecidable, answering an open problem of Lysenok and Ushakov (2021).
    We also construct a finitely generated abelian-by-cyclic group, different to that of Mishchenko and Treier (2017), in which the Knapsack Problem is undecidable.
    In contrast, we show that the problem of Coset Intersection is decidable in all finitely generated abelian-by-cyclic groups.
\end{abstract}

\vspace{0.5cm}


\vspace{0.5cm}
\noindent

\newpage
\setcounter{page}{1}

\section{Introduction and main results}
\subsection*{Linear equations with constraints}

In the first part of this paper we consider linear equations over the Laurent polynomial ring $\Z[X^{\pm}] = \Z[X, X^{-1}]$ with monomial constraints on certain variables.
Solving linear equations with additional constraints is widely studied in a much larger context.
Fix a commutative ring $R$ (such as $\Z, \Q, \Z[X]$ or $\Z[X^{\pm}]$) and a subset $S$ of $R$ (such as $\N, \Q_{>0}, \N[X]$ or $\{X^z \mid z \in \Z\}$).
Given $m \geq n \geq 1$, as well as elements $a_{ij} \in R, i = 0, 1, \ldots, m; j = 1, \ldots, d$, a central problem is to decide whether the system of linear equations
\begin{align}\label{eq:lin}
    a_{11} x_1 + \cdots + a_{n1} x_n + a_{n+1,1} x_{n+1} + \cdots + a_{m1} x_m & = a_{01}, \nonumber \\
    a_{12} x_1 + \cdots + a_{n2} x_n + a_{n+1, 2} x_{n+1} + \cdots + a_{m2} x_m & = a_{02}, \nonumber \\
    & \vdots \\
    a_{1d} x_1 + \cdots + a_{nd} x_n + a_{n+1, d} x_{n+1} + \cdots + a_{md} x_m & = a_{0d}, \nonumber
\end{align}
have solutions $x_1, \ldots, x_n \in S, x_{n+1}, \ldots, x_m \in R$.

For example, if one takes $R = \Z$ and its subset $S = \N$, then the problem of solving the system~\eqref{eq:lin} is equivalent to Integer Programming, and is thus NP-complete.
If we take $R = \Z$ and its subset $S = 2^{\N} \coloneqq \{2^n \mid n \in \N\}$, then solving the system~\eqref{eq:lin} can be formulated as deciding a fragment of the existential theory of \emph{Presburger arithmetic with power predicate}, which is decidable by a classic result of Semenov~\cite{semenov1980certain}.
It has since developed numerous applications and connections to automata theory~\cite{DBLP:conf/icalp/Benedikt0M23}.

In this paper, we focus on the case where $R$ is the ring $\Z[X^{\pm}]$ of univariate Laurent polynomials over integers. 
Most results on algorithmic problems can be applied interchangeably on the usual polynomial ring $\Z[X]$ and the Laurent polynomial ring $\Z[X^{\pm}]$.
For the purpose of subsequent applications on computational group theory (see below), we choose to state our results over the Laurent polynomial ring $\Z[X^{\pm}]$, although they can be easily adapted to $\Z[X]$.
All polynomial rings appearing in this paper will thus be Laurent polynomial rings.

Solving the system~\eqref{eq:lin} for $R = S = \Z[X^{\pm}]$ is a central topic in computational commutative algebra.
Numerous effective methods, such as the \emph{Gr\"{o}bner basis} (for modules)~\cite{eisenbud2013commutative, schreyer1980berechnung}, have been developed for deciding the existence of solutions for $R = S = \Z[X^{\pm}]$.
These algorithms also allows one to perform \emph{variable elimination}, and hence decide the existence of solutions for $R = \Z[X^{\pm}], S = \Z$.
However, adding positivity constraints for $R = \Z[X^{\pm}]$ can yield undecidability results.
Narendran~\cite{narendran1996solving} showed that it is undecidable whether the system~\eqref{eq:lin} has solutions for $R = \Z[X^{\pm}], S = \N[X^{\pm}]$.

In this paper we consider the problem of solving the system~\eqref{eq:lin} with ``monomial constraints'', that is, deciding the existence of solutions for $R = \Z[X^{\pm}]$ and $S = X^{\Z} \coloneqq \{X^z \mid z \in \Z\}$.
We prove this to be undecidable.
To express this problem more concisely, define $\bff_i = (a_{i1}, \ldots, a_{id}) \in \Z[X^{\pm}]^d, i = 0, 1, \ldots, m$, and denote by $Q$ the $\Z[X^{\pm}]$-module generated by $\bff_{n+1}, \ldots, \bff_{m}$. Then solving the system~\eqref{eq:lin} for $R = \Z[X^{\pm}], S = X^{\Z},$ is equivalent to finding solutions $z_1, \ldots, z_n \in \Z$ for the equation $X^{z_1} \bff_1 + \cdots + X^{z_n} \bff_n = \bff_0$ in the quotient module $\Z[X^{\pm}]^d/Q$.

\begin{theorem}\label{thm:linequndec}
    There exist an integer $n \in \N$, a finitely presented $\Z[X^{\pm}]$-module $\mA = \Z[X^{\pm}]^d/Q$, and elements $\bff_1, \ldots, \bff_n \in \mA$, such that the following problem is undecidable.

    \smallskip
    \noindent \textbf{Input:} an element $\bff_0 \in \mA$.
    
    \noindent \textbf{Question:} whether there exist $z_1, \ldots, z_n \in \Z$ such that
$
    X^{z_1} \bff_1 + \cdots + X^{z_n} \bff_n = \bff_0
$.
\end{theorem}

The proof of Theorem~\ref{thm:linequndec} will be given in Section~\ref{sec:lineq}. 
Our undecidability result does not depend on working over $\Z[X^{\pm}]$. Indeed, the reader can check that
the same argument shows the undecidability for finitely presented $\Q[X^{\pm}]$-modules and $\R[X^{\pm}]$-modules.
Our proof is based on an embedding of the Hilbert's tenth problem over the variables $z_1, \ldots, z_n$, and requires large $n$.
On the contrary, the case $n = 1$ is decidable.
This can be deduced from \emph{Noskov's Lemma} in commutative algebra:

\begin{restatable}[Corollary of Noskov's lemma~\cite{noskov1982conjugacy}]{theorem}{thmonemono}\label{thm:onemono}
    Given a finitely presented $\Z[X^{\pm}]$-module $\mA$ as well as two elements $\bff_0, \bff_1 \in \mA$, it is decidable whether $X^z \bff_1 = \bff_0$ has solution $z \in \Z$.
\end{restatable}
We will give a short deduction of Theorem~\ref{thm:onemono} in Appendix~\ref{app:proof}.

Our motivation behind studying linear equations over $\Z[X^{\pm}]$ with monomial constraint comes from a series of algorithmic problems in infinite groups.
In the second part of this paper, we focus on four decision problems: \emph{word equations}, \emph{quadratic word equations}, the \emph{Knapsack Problem} and \emph{Coset Intersection}, in the class of \emph{abelian-by-cyclic} groups.

\subsection*{Abelian-by-cyclic groups}
A group is called \emph{abelian-by-cyclic} if it admits an abelian normal subgroup $A$ such that the quotient group $G/A$ is isomorphic to $\Z$.
See Section~\ref{sec:prelim} for a structural description of abelian-by-cyclic groups.
Many well-studied groups, such as the \emph{wreath product} $\Z \wr \Z$ and the \emph{Baumslag-Solitar} group $\mathsf{BS}(1, p), p \geq 2$, fall into this category.
The groups $\Z \wr \Z$ and $\mathsf{BS}(1, p)$ can be intuitively understood as groups of $2 \times 2$ matrices over the ring of Laurent polynomials $\Z[X^{\pm}]$ and over the ring $\Z[1/p] = \{\frac{a}{p^n} \mid a \in \Z, n \in \N\}$:
\begin{align}
\Z \wr \Z & \cong \left\{ 
\begin{pmatrix}
        X^{b} & f \\
        0 & 1
\end{pmatrix}
\;\middle|\; f \in \Z[X^{\pm}], b \in \Z 
\right\}, \label{eq:defwr}
\\
\mathsf{BS}(1, p) & \cong \left\{ 
\begin{pmatrix}
        p^{b} & f \\
        0 & 1
\end{pmatrix}
\;\middle|\; f \in \Z[1/p], b \in \Z 
\right\}. \label{eq:defbs}
\end{align}
The element 
$
\begin{pmatrix}
        X^{b} & f \\
        0 & 1
\end{pmatrix}
$
(respectively
$
\begin{pmatrix}
        p^{b} & f \\
        0 & 1
\end{pmatrix}
$)
can be thought of as a Turing machine configuration whose tape cells contain letters in $\Z$ (resp.\ $\{0, 1, \ldots, p-1\}$), which correspond to the coefficients of the polynomial $f$ (resp.\ the base-$p$ expansion of $f$), while the head of the machine is positioned at the cell $b$.
Multiplication in the group corresponds to operating the machine by moving the head and adding elements of $\Z$ (resp.\ $\{0, 1, \ldots, p-1\}$, with carrying) to the cell. See~\cite{DBLP:conf/icalp/CadilhacCZ20} and~\cite{lohrey2015rational} for a complete description.
This analogy between group operations and Turing machines is shared among all abelian-by-cyclic groups, as the isomorphism $G/A \cong \Z$ gives the Turing machine-like structure, with $\Z$ representing the indices of the tape.

Abelian-by-cyclic groups have been extensively studied from the point of view of geometry and growth~\cite{farb2000asymptotic, hurtado2021global}, algorithmic problems~\cite{boler1976conjugacy}, and random walks~\cite{pittet2003random}.
Our motivation to study algorithmic problems in abelian-by-cyclic is two-fold.
On one hand, they are the simplest non-abelian groups, and thus one of simplest classes of groups where a number of algorithmic problems remain open.
On the other hand, the similarity of abelian-by-cyclic groups to operations on Turing machines might motivate the development of analogous methods in automata theory.

\subsection*{Equations over groups and quadratic equations}
The study of word equations over groups dates at least as far back as 1911, when Max Dehn proposed the word and conjugacy problems.
A major breakthrough in the study of equations over groups was made by Makanin in the 1980s.
Extending his seminal work on word equations in free semigroups~\cite{makanin1977problem}, Makanin showed that the question of whether a general system of equations over a free group has a solution is decidable~\cite{makanin1983equations}.
See~\cite{roman2012equations, DBLP:conf/icalp/DiekertE17, DBLP:conf/icalp/CiobanuE19, garreta2020diophantine, levine2022equations} for more recent developments.
Solving a system of equations over a groups $G$ can be formulated as the following decision problem.
Let $\mX = \{x_1, \ldots, x_n\}$ be a finite alphabet, define $\mX^{-1}$ as a new alphabet $\{x_1^{-1}, \ldots, x_n^{-1}\}$.
Suppose we are given a finite word $w$ over the alphabet $\mX \cup \mX^{-1} \cup G$ (the group $G$ might be infinite).
For any elements $g_1, \ldots, g_n \in G$, we can substitute each $x_i$ by $g_i$ (and $x_i^{-1}$ by $g_i^{-1}$) in $w$, and obtain a word over $G$.
We denote by $w(g_1, \ldots, g_n) \in G$ the product of this word.
The problem of solving a system of equations over a group $G$ can then be formulated as follows:

\smallskip
    \noindent \textbf{Input:} finite words $w_1, \ldots, w_t$ over the alphabet $\mX \cup \mX^{-1} \cup G$.
    
    \noindent \textbf{Question:} whether there exist $g_1, \ldots, g_n \in G$, such that
\[
    w_1(g_1, \ldots, g_n) = \cdots = w_t(g_1, \ldots, g_n) = e.
\]

Here, $e$ denotes the neutral element of $G$.
For example, the \emph{conjugacy problem}, which asks for given $h, h' \in G$, whether there exists $g \in G$ such that $ghg^{-1} = h'$, can be considered as a special case of solving an equation with $\mX = \{x_1\}, w_1 = x_1 h x_1^{-1} h'^{-1}$.

One class of equations over groups that generated much interest is the class of \emph{quadratic equations}: these consist of only words where for each $i$, the numbers of occurrences of $x_i$ and $x_i^{-1}$ sum up to two (such as the word $w_1 = x_1 h x_1^{-1} h'^{-1}$).
Apart from being a generalization of the conjugacy problem, quadratic equations have also been observed to have tight connections with the theory of compact surfaces~\cite{culler1981using, schupp1980quadratic}.
Recently, Mandel and Ushakov~\cite{mandel2023quadratic} showed that solving systems of quadratic equations is decidable and NP-complete in the Baumslag-Solitar group $\mathsf{BS}(1, p)$, while Ushakov and Weiers~\cite{ushakov2023quadratic} showed that that solving systems of quadratic equations is NP-complete in the wreath product $\Z_2 \wr \Z$ (commonly known as the \emph{lamplighter group}).
Kharlampovich, L\'{o}pez and Miasnikov~\cite{kharlampovich2020diophantine} showed that solving systems of quadratic equations is decidable in the wreath product $A \wr \Z$, where $A$ is any finitely generated abelian group.
They left as an open problem\footnote{In the
published version of Kharlampovich, L\'{o}pez and Miasnikov's paper~\cite{kharlampovich2020diophantine}, it was claimed that solving a general system of equations is \emph{decidable} in $A \wr \Z$, with $A$ finitely generated abelian. The proof unfortunately contained a gap. In the corrected version, the problem was again listed as open.} whether solving a \emph{general} system of equations is decidable in $A \wr \Z$.
In this paper, we provide a negative answer to this open problem:

\begin{restatable}{theorem}{thmsyseq}\label{thm:syseq}
    Solving a system of equations is undecidable in $\Z \wr \Z$.
\end{restatable}

A fortiori, the existential theory of the group $\Z \wr \Z$ is undecidable.
We also show that the decidability of solving systems of quadratic equations in $A \wr \Z$ and $\mathsf{BS}(1, p)$ cannot be extended to \emph{all} abelian-by-cyclic groups:

\begin{restatable}{theorem}{thmquadeq}\label{thm:quadeq}
    There exists a finitely generated abelian-by-cyclic group $\mA \rtimes \Z$, as well as its elements $h_1, \ldots, h_m$, such that the following problem is undecidable.
    
    \smallskip
    \noindent \textbf{Input:} an element $h_0 \in \mA \rtimes \Z$.
    
    \noindent \textbf{Question:} whether there exist $g_0, g_1, \ldots, g_m \in \mA \rtimes \Z$, such that
\[
    (g_0 h_0 g_0^{-1}) (g_1 h_1 g_1^{-1}) \cdots (g_m h_m g_m^{-1}) = e.
\]    
    A fortiori, solving (a single) quadratic equation is undecidable in the abelian-by-cyclic group $\mA \rtimes \Z$.
\end{restatable}
Note that the quadratic equation appearing in Theorem~\ref{thm:quadeq} has the special form of a \emph{spherical equation}~\cite{lysenok2021orientable}.
This shows that solving spherical equations is undecidable in finitely generated metabelian groups, providing a negative answer to an open problem of Lysenok and Ushakov~\cite{lysenok2021orientable}.

Theorem~\ref{thm:syseq} and \ref{thm:quadeq} are our main group theory results, and will be proven in Subsections~\ref{subsec:ZwrZ} and \ref{subsec:quadeq}, respectively.

\subsection*{Knapsack Problem}
In~\cite{myasnikov2015knapsack}, Myasnikov, Nikolaev, and Ushakov began the investigation of classical optimization problems, which are formulated over the integers, for noncommutative groups.
The \emph{Knapsack Problem} was introduced for a finitely generated group $G$ among other problems. 
The input for the Knapsack Problem is a sequence of group elements $g_1, \ldots , g_n, g \in G$, and the question is whether there exist non-negative integers $z_1, \ldots, z_n \in \N$ such that $g_1^{z_1} g_2^{z_2} \cdots g_n^{z_n} = g$.

On one hand, by taking $G$ to be the group of integers $\Z$, we recover a variant of the classical knapsack problem.
On the other hand, the Knapsack Problem for an arbitrary group $G$ can be seen as an essential special case of the \emph{Rational Subset Membership Problem} (whether a given regular expression over $G$ contains the neutral element).
Indeed, solving the Knapsack Problem often provides an important first step towards solving Rational Subset Membership~\cite{bodart2024membership}.
The Knapsack Problem in non-commutative groups has lately received much attention from computational group theorists.
Figelius, Ganardi, K\"{o}nig, Lohrey and Zetzsche showed decidability and and NP-completeness of the Knapsack Problem in wreath products $G \wr \Z$, where $G$ is any non-trivial finitely generated abelian group~\cite{DBLP:conf/stacs/GanardiKLZ18}, or any finite nilpotent group~\cite{DBLP:conf/icalp/FigeliusGLZ20}. 
Nevertheless, decidability of the Knapsack Problem does not extend to general abelian-by-cyclic groups. Mishchenko and Treier~\cite{mishchenko2017knapsack} showed that there exist abelian-by-cyclic groups with undecidable Knapsack Problem.
In Section~\ref{sec:KP}, we apply Theorem~\ref{thm:linequndec} give a different and more direct proof of Mishchenko and Treier's result:

\begin{restatable}{theorem}{thmknapsack}\label{thm:knapsack}
    There exists a finitely generated abelian-by-cyclic group $\mA \rtimes \Z$, as well as its elements $g_1, \ldots, g_m$, such that the following problem is undecidable.
    
    \smallskip
    \noindent \textbf{Input:} an element $g \in \mA \rtimes \Z$.
    
    \noindent \textbf{Question:} whether there exist $z_1, \ldots, z_m \in \N$ such that
$
    g_1^{z_1} g_2^{z_2} \cdots g_m^{z_m} = g
$.
    \smallskip        
    
    A fortiori, the Knapsack Problem is undecidable in the abelian-by-cyclic group $\mA \rtimes \Z$.
\end{restatable}

\subsection*{Coset Intersection}
The last problem we study is \emph{Coset Intersection}.
Given a finite subset $S$ of a group $G$, denote by $\gen{S}$ the subgroup generated by $S$, and let $h \gen{S}$ denote the coset $\{hg \mid g \in \gen{S}\}$.
The input of the Coset Intersection problem is two finite subsets $\mG, \mH \subset G$ as well as an element $h \in G$, and the question is whether the intersection $\gen{\mG} \cap h \gen{\mH}$ is empty.

The problem of Coset Intersection is motivated by numerous other areas such as Graph Isomorphism~\cite{luks1982isomorphism}, vector reachability~\cite{potapov2019vector} and automata theory~\cite{delgado2018intersection}.
Coset Intersection is related to the conjugacy problem and quadratic equations, as solving a system of equations of the form $w = x_1 h x_1^{-1} h'^{-1}$ boils down to deciding Coset Intersection.
At the same time, Coset Intersection is a much more tractable special case of Rational Subset Membership, compared to the Knapsack Problem.
Recent results by Lohrey, Steinberg and Zetzsche~\cite{lohrey2015rational} showed decidability of Rational Subset Membership in $\Z_p \wr \Z,\; p \geq 2$.
This result has been extended to $\mathsf{BS}(1, p),\; p \geq 2,$ by Cadilhac, Chistikov and Zetzsche~\cite{DBLP:conf/icalp/CadilhacCZ20}.
In Section~\ref{sec:inter}, we show that Coset Intersection is in fact decidable in \emph{all} finitely generated abelian-by-cyclic groups:

\begin{restatable}{theorem}{thmcoset}\label{thm:coset}
    Given as input a finitely generated abelian-by-cyclic group $\mA \rtimes \Z$, as well as two finite sets of elements $\mG = \{g_1, \ldots, g_K\}, \mH = \{h_1, \ldots, h_M\}$ and $h \in \mA \rtimes \Z$, it is decidable whether $\gen{\mG} \cap h \gen{\mH} = \emptyset$.
\end{restatable}


Table~\ref{tbl:stateofart} summarizes our results for abelian-by-cyclic groups in the context of the current state of art.
The question marks denote problems whose decidability status is unknown. The decidability status for empty blocks are subsumed by results in the same row or column.

\begin{table}[h!]
\centering
\begin{tabular}{ | m{1.5cm} | m{2.4cm}| m{2.2cm}| m{2.4cm} | m{2.0cm} | m{2.8cm} | } \hline
    & Quadratic Equations & General \; Equations & Knapsack Problem & Coset \quad\quad\; Intersection & Rational \;\;\quad Subset\\
  \hline
  \hline
  $\mathsf{BS}(1, p)$ & NP-complete \cite{mandel2023quadratic} & ? & NP-complete \cite{lohrey2020knapsack} & & PSPACE-comp.\ \cite{DBLP:conf/icalp/CadilhacCZ20} \\
  \hline
  $\Z_p \wr \Z$ & NP-complete \cite{ushakov2023quadratic} & ? & NP-complete \cite{DBLP:conf/stacs/GanardiKLZ18} & & decidable \cite{lohrey2015rational}\\
  \hline
  $\Z \wr \Z$ & decidable \cite{kharlampovich2020diophantine} & undecidable\textsuperscript{\textdagger} & NP-complete \cite{DBLP:conf/stacs/GanardiKLZ18} &  & undecidable \cite{lohrey2015rational} \\
  \hline
  abelian-by-cyclic & undecidable\textsuperscript{\textdagger} & & undecidable$^{*}$ \cite{mishchenko2017knapsack} & decidable\textsuperscript{\textdagger} & \\
  \hline
\end{tabular}
\caption{\label{tbl:stateofart} \quad \textsuperscript{\textdagger} $=$ our new results, \quad $^* =$ our alternative construction}
\end{table}

\section{Preliminaries}\label{sec:prelim}

\subsection*{Laurent polynomial ring and modules}

A (univariate) \emph{Laurent polynomial} with coefficients over $\Z$ is an expression of the form
\[
f = \sum_{i = p}^q a_i X^i, \quad \text{where $p, q \in \Z$ and $a_i \in \Z, i = p, p+1, \ldots, q$.}
\]
The set of all Laurent polynomials with coefficients over $\Z$ forms a ring and is denoted by $\Z[X^{\pm}]$.

Let $R$ be a commutative ring.
An $R$-module is defined as an abelian group $(M, +)$ along with an operation $\cdot \;\colon R \times M \rightarrow M$ satisfying $f \cdot (a+b) = f \cdot a + f \cdot b$, $(f + g) \cdot a = f \cdot a + g \cdot a$, $fg \cdot a = f \cdot (g \cdot a)$ and $1 \cdot a = a$.
We will denote by $\bzer$ the neutral element of an $R$-module $M$.
If $N$ and $N'$ are $R$-submodule of $M$, then $N + N' \coloneqq \{n + n' \mid n \in N, n' \in N'\}$ is again an $R$-submodule of $M$.

For any $d \in \N$, the direct power $\Z[X^{\pm}]^d$ is a $\Z[X^{\pm}]$-module by $g \cdot (f_1, \ldots, f_d) \coloneqq (gf_1, \ldots, gf_d)$.
Throughout this paper, we use the bold symbol $\bff$ to denote a vector $(f_1, \ldots, f_d) \in \Z[X^{\pm}]^d$.
Given $\bff_1, \ldots, \bff_m \in \Z[X^{\pm}]^d$, we say they \emph{generate} the $\Z[X^{\pm}]$-module 
\[
\sum_{i=1}^m \Z[X^{\pm}] \cdot \bff_i \coloneqq \left\{\sum_{i=1}^m p_i \cdot \bff_i \;\middle|\; p_1, \ldots, p_m \in \Z[X^{\pm}] \right\}.
\]
Given two $\Z[X^{\pm}]$-submodules $N, M$ of $\Z[X^{\pm}]^d$ such that $N \subseteq M$, we can define the quotient $M/N \coloneqq \{\overline{\bm} \mid \bm \in M\}$ where $\overline{\bm_1} = \overline{\bm_2}$ if and only if $\bm_1 - \bm_2 \in N$.
This quotient is also a $\Z[X^{\pm}]$-module.
We say that a $\Z[X^{\pm}]$-module $\mA$ is \emph{finitely presented} if it can be written as a quotient $M/N$ for two submodules $M, N$ of $\Z[X^{\pm}]^d$ for some $d \in \N$, where both $M$ and $N$ are generated by finitely many elements.
Such a pair $(M, N)$, given by their respective generators, is called a \emph{finite presentation} of $\mA$.
The element $\overline{\bm}$ of $\mA$ is effectively represented by $\bm \in \Z[X^{\pm}]^d$, this representation is unique modulo $N$.
In this paper, we will often write $\bm$ instead of $\overline{\bm}$, when the context is clear that this represents an element of the quotient $\mA = M/N$.

\subsection*{Abelian-by-cyclic groups}
We now formally define abelian-by-cyclic groups.
\begin{definition}
    A group $G$ is called \emph{abelian-by-cyclic} if it admits an abelian normal subgroup $\mA$ such that $G/\mA \cong \Z$.
\end{definition}


The abelian subgroup $\mA$ admits a natural $\Z[X^{\pm}]$-module structure in the following sense.
It is a classic result~\cite[p.17]{boler1976conjugacy} that every finitely generated abelian-by-cyclic group $G$ can be written as a \emph{semidirect product} $\mA \rtimes \Z$:
\begin{equation}
    \mA \rtimes \Z \coloneqq \left\{ (\ba, z) \;\middle|\; \ba \in \mA, z \in \Z \right\},
\end{equation}
where $\mA$ is a finitely presented $\Z[X^{\pm}]$-module.
The group law in $\mA \rtimes \Z$ is defined by
\[
(\ba, z) \cdot (\ba', z') = (\ba + X^z \cdot \ba', z + z'), \quad (\ba, z)^{-1} = (- X^{-z} \cdot \ba, -z).
\]
The neutral element of $\mA \rtimes \Z$ is $(\bzer, 0)$.
Intuitively, the element $(\ba, z)$ is analogous to a $2 \times 2$ matrix
$
\begin{pmatrix}
X^{z} & \ba \\
0 & 1 \\
\end{pmatrix}
$, where group multiplication is represented by matrix multiplication.
We naturally identify $\mA$ with the normal subgroup $\left\{ (\ba, 0) \;\middle|\; \ba \in \mA\right\}$ of $\mA \rtimes \Z$.
In particular, the quotient $\left( \mA \rtimes \Z \right)/\mA$ is isomorphic to $\Z$, so $\mA \rtimes \Z$ is indeed abelian-by-cyclic.

\begin{example}
    Take $\mA \coloneqq \Z[X^{\pm}]$ considered as a $\Z[X^{\pm}]$-module, then we recover the definition~\eqref{eq:defwr} of the wreath product $\Z \wr \Z = \Z[X^{\pm}] \rtimes \Z$.
    If $\mA \coloneqq \Z_p[X^{\pm}] = \Z[X^{\pm}]/\left(\Z[X^{\pm}] \cdot p\right)$, then we obtain the lamplighter group $\Z_p \wr \Z = \Z_p[X^{\pm}] \rtimes \Z$.
    In general, if $A$ is any finitely generated abelian group, and let $A[X^{\pm}]$ denote the set of Laurent polynomials with coefficients in $A$, then $A[X^{\pm}]$ is a finitely presented $\Z[X^{\pm}]$-module (but not necessarily a ring), and the wreath product $A \wr \Z$ is defined as the semidirect product $A[X^{\pm}] \rtimes \Z$.

    If we take $\mA \coloneqq \Z[1/p] = \Z[X^{\pm}]/\left(\Z[X^{\pm}] \cdot (X - p)\right)$, then we recover the definition~\eqref{eq:defbs} of the Baumslag-Solitar group $\mathsf{BS}(1, p) = \Z[1/p] \rtimes \Z$.
\end{example}

Throughout this paper, a finitely generated abelian-by-cyclic group $G$ is always represented as the semidirect product $\mA \rtimes \Z$, where $\mA$ is a $\Z[X^{\pm}]$-module given by a finite presentation.

\section{Linear equation with monomial constraints}\label{sec:lineq}

In this section, we prove undecidability of Theorem~\ref{thm:linequndec} by embedding Hilbert's tenth problem.
For $f, g \in \Z[X^{\pm}]$, we write $f \mid g$ if there exists $h \in \Z[X^{\pm}]$ such that $g = fh$.
We show undecidability of Theorem~\ref{thm:linequndec}, even for the following special case:

\begin{proposition}\label{prop:linequndec}
    There exist $k, n \in \N$, polynomials $p_1, \ldots, p_k \in \{0, (X-1)^2, (X-1)^3\}$, and $f_{ij} \in \Z[X^{\pm}]$, $i = 1, \ldots, n; \; j = 1, \ldots, k$, such that the following problem is undecidable:

    \smallskip
    \noindent \textbf{Input:} polynomials $f_{0j} \in \Z[X^{\pm}]$, $j = 1, \ldots, k$.
    
    \noindent \textbf{Question:} 
    whether there exist $z_1, \ldots, z_n \in \Z$ that satisfy the following system:
    \begin{align}\label{eq:sysmono}
        p_1 & \mid X^{z_1} f_{11} + \cdots + X^{z_n} f_{n1} - f_{01}, \nonumber\\
        p_2 & \mid X^{z_1} f_{12} + \cdots + X^{z_n} f_{n2} - f_{02}, \nonumber\\
        & \vdots \nonumber\\
        p_k & \mid X^{z_1} f_{1k} + \cdots + X^{z_n} f_{nk} - f_{0k}.
    \end{align}
    In particular, when $p_i = 0$, the expression $p_i \mid F$ means $F = 0$.
\end{proposition}

Note that Proposition~\ref{prop:linequndec} will immediately yield Theorem~\ref{thm:linequndec}. Indeed, let $p_1, \ldots, p_k \in \Z[X^{\pm}]$ be as in Proposition~\ref{prop:linequndec}, and let $\mA$ be the finitely presented $\Z[X^{\pm}]$-module $\Z[X^{\pm}]^k/Q$, where
\[
Q \coloneqq \{(g_1 p_1, g_2 p_2, \ldots, g_k p_k) \mid g_1, \ldots, g_k \in \Z[X^{\pm}]\}.
\]
Then for $\bff_0 = (f_{01}, \ldots, f_{0k}), \bff_1 = (f_{11}, \ldots, f_{1k}), \ldots, \bff_n = (f_{n1}, \ldots, f_{nk}) \in \mA$, we have $X^{z_1} \bff_1 + \cdots + X^{z_n} \bff_n = \bff_0$ if and only if the system~\eqref{eq:sysmono} is satisfied. Therefore, Proposition~\ref{prop:linequndec} implies Theorem~\ref{thm:linequndec}.
We now start proving Proposition~\ref{prop:linequndec}.

\begin{lemma}\label{lem:sqaure}
    Suppose $z_1, z_2, z_3 \in \Z$. We have
    \begin{equation}\label{eq:div3}
    (X-1)^3 \mid X^{z_1} + X^{z_2} (1-X) + X^{z_3} + (X - 3)
    \end{equation}
    if and only if $z_2 = z_1^2$ and $z_3 = - z_1$.
\end{lemma}

\begin{proof}
    Note that $(X-1)^3 \mid f$ if and only if $f(1) = f'(1) = f''(1) = 0$, where $f'$ and $f''$ denote the first and second derivatives of $f$. Taking $f = X^{z_1} + X^{z_2} (1-X) + X^{z_3} + (X - 3)$, then $f(1) = f'(1) = f''(1) = 0$ is equivalent to
    \[
    0 = z_1 + z_3 = z_1^2 - z_1 - 2 z_2 + z_3^2 - z_3 = 0,
    \]
    which is equivalent to $z_2 = z_1^2, z_3 = - z_1$.
\end{proof}

\begin{lemma}\label{lem:sum}
    Suppose $z_1, z_2, z_3 \in \Z$. We have
    \[
    (X-1)^2 \mid X^{z_1} + X^{z_2} - X^{z_3} - 1
    \]
    if and only if $z_3 = z_1 + z_2$.
\end{lemma}

\begin{proof}
    Note that $(X-1)^2 \mid f$ if and only if $f(1) = f'(1) = 0$. Taking $f = X^{z_1} + X^{z_2} - X^{z_3} - 1$, then $f(1) = f'(1) = 0$ is equivalent to
    $
    0 = z_1 + z_2 - z_3 = 0
    $.
\end{proof}

\begin{lemma}\label{lem:product}
    Suppose $z_i \in \Z, i = 1, 2, 3$. There exist $z_4, \ldots, z_{12} \in \Z$ satisfying
    \begin{align}
        (X-1)^3 & \mid X^{z_1} + X^{z_4} (1-X) + X^{z_{10}} + (X - 3), \label{eq:1}\\
        (X-1)^3 & \mid X^{z_2} + X^{z_5} (1-X) + X^{z_{11}} + (X - 3), \label{eq:2}\\
        (X-1)^2 & \mid X^{z_4} + X^{z_5} - X^{z_6} - 1, \label{eq:3} \\
        (X-1)^2 & \mid X^{z_1} + X^{z_2} - X^{z_7} - 1, \label{eq:4} \\
        (X-1)^3 & \mid X^{z_7} + X^{z_8} (1-X) + X^{z_{12}} + (X - 3), \label{eq:5}\\
        (X-1)^2 & \mid X^{z_3} + X^{z_3} - X^{z_9} - 1, \label{eq:6}\\
        (X-1)^2 & \mid X^{z_6} + X^{z_9} - X^{z_8} - 1, \label{eq:7}
    \end{align}
    if and only if $z_3 = z_1 z_2$.
\end{lemma}
\begin{proof}
    We apply Lemmas~\ref{lem:sqaure} or \ref{lem:sum} to each of the equations~\eqref{eq:1}-\eqref{eq:7}.
    Equations~\eqref{eq:1}, \eqref{eq:2} and \eqref{eq:3} are equivalent to
    \[
    z_4 = z_1^2, z_5 = z_2^2, z_6 = x_4 + x_5 = z_1^2 + z_2^2,
    \]
    and $z_{10} = - z_1, z_{11} = - z_2$.
    Equations~\eqref{eq:4} and \eqref{eq:5} are equivalent to
    \[
    z_7 = z_1 + z_2, z_8 = z_7^2 = z_1^2 + z_2^2 + 2 z_1 z_2,
    \]
    and $z_{12} = - z_7 = - z_1 - z_2$.
    Finally, Equations~\eqref{eq:6} and \eqref{eq:7} are equivalent to
    \[
    z_9 = z_3 + z_3, z_8 = z_6 + z_9 = z_6 + z_3 + z_3.
    \]
    Together with $z_8 = z_1^2 + z_2^2 + 2 z_1 z_2$ and $z_6 = z_1^2 + z_2^2$, this yields $z_3 = z_1 z_2$.
\end{proof}

We now combine Lemma~\ref{lem:sum} and \ref{lem:product} to prove Proposition~\ref{prop:linequndec} (and consequently, Theorem~\ref{thm:linequndec}).

\begin{proof}[Proof of Proposition~\ref{prop:linequndec} and Theorem~\ref{thm:linequndec}]
    By Matiyasevich’s proof of the undecidability of Hilbert’s tenth problem~\cite{matiyasevich1993hilbert}, there exists a polynomial $P_{uni} \in \Z[X_1, \ldots, X_d]$, such that the following problem is undecidable:

    \smallskip
    \noindent \textbf{Input:} an integer $a \in \Z$.
    
    \noindent \textbf{Question:} whether $P_{uni}(x_1, \ldots, x_d) = a$ has a solution $(x_1, \ldots, x_d) \in \Z^d$.
    \smallskip

    To embed the Hilbert's tenth problem into Equation~\eqref{eq:sysmono}, we will construct $k, n \in \N$, polynomials $p_1, \ldots, p_{k-1} \in \{0, (X-1)^2, (X-1)^3\}$, $f_{ij} \in \Z[X^{\pm}]$, $i = 1, \ldots, n; \; j = 1, \ldots, k$, and $f_{0j} \in \Z[X^{\pm}]$, $j = 1, \ldots, k-1$, such that $P_{uni}(z_1, \ldots, z_d) = a$ if and only if there exist $z_{d+1}, \ldots, z_n \in \Z$ such that
    \begin{align}\label{eq:sysmono2}
        p_1 & \mid X^{z_1} f_{11} + \cdots + X^{z_n} f_{n1} - f_{01}, \nonumber\\
        p_2 & \mid X^{z_1} f_{12} + \cdots + X^{z_n} f_{n2} - f_{02}, \nonumber\\
        & \vdots \nonumber\\
        p_{k-1} & \mid X^{z_1} f_{1, k-1} + \cdots + X^{z_n} f_{n, k-1} - f_{0, k-1}, \nonumber\\
        0 & \mid X^{z_n} - X^a.
    \end{align}

    First we rewrite the equation $P_{uni}(z_1, \ldots, z_d) = z_n$ into a system of polynomial equations, each of the form $z_k = z_i z_j$ or $z_k = z_i + z_j$ or $z_k = b, b \in \Z$, with $i, j, k \in \{1, 2, \ldots, n\}$.
    For example, ``$z_1^3 + 2 z_1 z_2 + 7 = z_n$'' can be rewritten as ``$z_3 = z_1 z_1, z_4 = z_1 z_3, z_5 = z_1 z_2, z_6 = z_5 + z_5, z_7 = z_4 + z_6, z_8 = 7, z_n = z_7 + z_8$''.

    Then, we use Lemma~\ref{lem:product} and \ref{lem:sum} respectively to express the equations $z_k = z_i z_j$ and $z_k = z_i + z_j$, and use $0 \mid X^{z_k} - X^b$ to express $z_k = b$. Finally, we add $0 \mid X^{z_n} - X^a$ to the system to fully express the equation $P_{uni}(z_1, \ldots, z_d) = a$.
    The resulting system have the form~\eqref{eq:sysmono2} and is equivalent to $P_{uni}(z_1, \ldots, z_d) = a$ (the variables $z_{d+1}, \ldots, z_n$ will be functions over $z_1, \ldots, z_d$ and do not impose extra constraints).
    
    Notice that the input $a \in \Z$ only appear once as $X^a$ in the system~\eqref{eq:sysmono2}, which corresponds to $f_{0k} = X^a$ in the system~\eqref{eq:sysmono}. Therefore the decision problem stated in Proposition~\ref{prop:linequndec} is undecidable.
    As explained after the statement of Proposition~\ref{prop:linequndec}, the correctness of Theorem~\ref{thm:linequndec} directly follows.
\end{proof}

\section{Equations over abelian-by-cyclic groups}\label{sec:reduction}

\subsection{Quadratic equation undecidable in an abelian-by-cyclic group}\label{subsec:quadeq}
\hfill

\noindent In this subsection, we prove Theorem~\ref{thm:quadeq}, showing that solving quadratic equations is undecidable in the abelian-by-cyclic group $\mA \rtimes \Z$, where $\mA$ is the $\Z[X^{\pm}]$-module constructed in Theorem~\ref{thm:linequndec}.
Our strategy is a reduction from Theorem~\ref{thm:linequndec}, by embedding the equation $X^{z_1} \bff_1 + \cdots + X^{z_n} \bff_n = \bff_0$ into a quadratic equation over $\mA \rtimes \Z$.
The key observation is that ``multiplying by $X^z$'' can be expressed by conjugacy:
$
    (\ba, z) (\bff, 0) (\ba, z)^{-1} = (X^{z} \bff, 0)
$.

\thmquadeq*
\begin{proof}
    We will reduce the problem from Theorem~\ref{thm:linequndec}.
    Let $\mA$ be the finitely presented $\Z[X^{\pm}]$-module constructed in the theorem.
    Take $h_1 = (\bff_1, 0), h_2 = (\bff_2, 0) \ldots, h_n = (\bff_n, 0)$, where $\bff_1, \ldots, \bff_n \in \mA$ are the elements constructed in Theorem~\ref{thm:linequndec}.
    For an instance of the decision problem in Theorem~\ref{thm:linequndec} with input $\bff_0 \in \mA$,
    we take the input $h_0 = (- \bff_0, 0)$ in the above decision problem.
    We show that there exist $g_0, g_1, \ldots, g_n \in \mA \rtimes \Z$, such that
    \begin{equation}\label{eq:qe}
        (g_0 h_0 g_0^{-1}) (g_1 h_1 g_1^{-1}) \cdots (g_n h_n g_n^{-1}) = e,
    \end{equation}
    if and only if there exist $z_1, \ldots, z_n \in \Z$, such that
    \begin{equation}\label{eq:qelinmono}
         X^{z_1} \bff_1 + \cdots + X^{z_n} \bff_n = \bff_0.
    \end{equation}

    If there exist $g_0, g_1, \ldots, g_n \in \mA \rtimes \Z$ that satisfy \eqref{eq:qe}, write $g_0 = (\ba_0, b_0), g_1 = (\ba_1, b_1), \ldots$, $g_n = (\ba_n, b_n)$.
    Then for $i = 1, \ldots, n$, we have 
    \[
    g_i h_i g_i^{-1} = (\ba_i, b_i) (\bff_i, 0) (- X^{-b_i} \ba_i, - b_i) = (X^{b_i} \bff_i, 0),
    \]
    and $g_0 h_0 g_0^{-1} = (- X^{b_0} \bff_0, 0)$.
    Therefore, Equation~\eqref{eq:qe} is equivalent to
    \begin{align*}
        \left( - X^{b_0} \bff_0 + X^{b_1} \bff_1 + \cdots + X^{b_n} \bff_n, 0 \right) = \left(\bzer, 0 \right),
    \end{align*}
    which is equivalent to
    \[
    X^{b_1 - b_0} \bff_1 + \cdots + X^{b_n - b_0} \bff_n = \bff_0.
    \]
    Thus, the integers $z_1 = b_1 - b_0, z_2 = b_2 - b_0, \ldots, z_n = b_n - b_0$ satisfy Equation~\eqref{eq:qelinmono}.

    Conversely, if Equation~\eqref{eq:qelinmono} has solutions $z_1, \ldots, z_m \in \Z$, take $g_0 = (\bzer, 0), g_1 = (\bzer, z_1), \ldots, g_n = (\bzer, z_n) \in \mA \rtimes \Z$.
    Then as above we have $g_0 h_0 g_0^{-1} = (- \bff_0, 0)$ and $g_i h_i g_i^{-1} = (X^{z_i} \bff_i, 0)$ for $i = 1, \ldots, n$.
    Therefore $g_0, g_1, \ldots, g_n$ satisfy Equation~\eqref{eq:qe}.

    By Theorem~\ref{thm:linequndec}, it is undecidable whether Equation~\eqref{eq:qelinmono} has integer solutions. Therefore, it is also undecidable whether Equation~\eqref{eq:qe} has solutions in $\mA \rtimes \Z$.
\end{proof}

\subsection{General system of equations undecidable in $\Z \wr \Z$}\label{subsec:ZwrZ}
\hfill

\noindent Recall that $\Z \wr \Z$ is defined as $\Z[X^{\pm}] \rtimes \Z$. 
Its elements are $(f, a)$ where $f \in \Z[X^{\pm}], a \in \Z$.
In this subsection we show that solving a general system of equations is undecidable in $\Z \wr \Z$ (Theorem~\ref{thm:syseq}). 

The difficulty of proving Theorem~\ref{thm:syseq}, compared to Theorem~\ref{thm:quadeq}, is that we can no longer rely on the structure of the abelian-by-cyclic group to embed the information on the module $\mA$ from Theorem~\ref{thm:linequndec}.
We therefore need to dig deeper, and prove Theorem~\ref{thm:syseq} by reducing it from Proposition~\ref{prop:linequndec} instead of Theorem~\ref{thm:linequndec}.
Recall that Proposition~\ref{prop:linequndec} states it is undecidable whether a system of equations of the form $p_j \mid X^{z_1} f_{1j} + \cdots + X^{z_n} f_{nj} - f_{0j}, \; j = 1, \ldots, k,$ has solutions $z_1, \ldots, z_n \in \Z$.
Furthermore, the divisors $p_j$ have either the form $0$, $(X-1)^2$, or $(X-1)^3$.
The key idea of embedding these equations into equations over $\Z \wr\Z$ is that divisibility by $(X-1)^k$ can be expressed by an equation over $\Z \wr\Z$ containing the \emph{$k$-th commutator words}.
For an alphabet $\mY = \{x, y, z, w\}$, define recursively the following commutator words over $\mY \cup \mY^{-1}$:
\begin{align*}
    [x, y] & \coloneqq x^{-1} y^{-1} xy, \\
    [[x, y], z] & \coloneqq [x^{-1}y^{-1}xy, z] = y^{-1} x^{-1} y x z^{-1} x^{-1} y^{-1} xyz, \\
    [[[x, y], z], w] & \coloneqq [[x^{-1}y^{-1}xy, z], w] = z^{-1} y^{-1} x^{-1} y x z x^{-1} y^{-1} xy w^{-1} y^{-1} x^{-1} y x z^{-1} x^{-1} y^{-1} xyz w.
\end{align*}

\begin{lemma}\label{lem:comms}
Let $(f, b) \in \Z \wr \Z$, let $e = (0, 0)$ denote the neutral element of $\Z \wr \Z$.
\begin{enumerate}[nosep, label=(\arabic*)]
    \item There exists $x, y \in \Z \wr \Z$ such that $(f, b) \cdot [x, y] = e$, if and only if $(X-1) \mid f$ and $b = 0$.
    \item There exists $x, y, z \in \Z \wr \Z$ such that $(f, b) \cdot [[x, y], z] = e$, if and only if $(X-1)^2 \mid f$ and $b = 0$.
    \item There exists $x, y, z, w \in \Z \wr \Z$ such that $(f, b) \cdot [[[x, y], z], w] = e$, if and only if $(X-1)^3 \mid f$ and $b = 0$.
\end{enumerate}
\end{lemma}
\begin{proof}
    (1) Write $x = (f_1, a_1), y = (f_2, a_2)$. By direct computation, 
    \[
    [x, y] = (f_1, a_1)^{-1} (f_2, a_2)^{-1} (f_1, a_1) (f_2, a_2) = ((1 - X^{a_2})X^{- a_1 - a_2} f_1 + (X^{a_1} - 1)X^{- a_1 - a_2} f_2, 0).
    \]
    If $(f, b) \cdot [x, y] = e$ for some $x = (f_1, a_1), y = (f_2, a_2)$, then $b = 0$ and $f = - (1 - X^{a_2})X^{- a_1 - a_2} f_1 - (X^{a_1} - 1)X^{- a_1 - a_2} f_2$. Since $X - 1 \mid 1 - X^{a}$ for all $a \in \Z$, we have $(X-1) \mid f$.
    Conversely, if $b = 0$ and $(X - 1) \mid f$, then one can take $x = (0, -1), y = (\frac{f}{X - 1}, 0)$, so we have $[x, y] = (-f, 0)$ and thus $(f, b) \cdot [x, y] = e$.

    (2) By (1), for any $x, y \in \Z \wr \Z$, we can write $[x, y]$ as $(F, 0)$ for some $(X - 1) \mid F$. And conversely for any $(X - 1) \mid F$ we can find $x, y \in \Z \wr \Z$ such that $[x, y] = (F, 0)$.
    Write $z = (f_3, a_3)$, then 
    \[
    [[x, y], z] = [(F, 0), (f_3, a_3)] = ((1 - X^{a_3})X^{- a_3} F, 0).
    \]
    Therefore, if there are $x, y, z \in \Z \wr \Z$ such that $(f, b) \cdot [[x, y], z] = e$, then $b = 0$ and $f = -(1 - X^{a_3})X^{- a_3} F$. Since $(X-1) \mid F$ and $(X-1) \mid 1 - X^{a_3}$, we have $(X-1)^2 \mid f$.
    Conversely, if $b = 0$ and $(X - 1)^2 \mid f$, then by (1) we can find $x, y \in \Z \wr \Z$ such that $[x, y] = (- \frac{f}{X-1}, 0)$. We then take $z = (0, -1)$, so $[[x, y], z] = (- f, 0)$ and $(f, b) \cdot [[x, y], z] = e$.

    (3) follows from (2) the same way (2) follows from (1).
\end{proof}

\thmsyseq*
\begin{proof}
    We will embed the decision problem from Proposition~\ref{prop:linequndec} into the problem of finding solutions to a system of equations in $\Z \wr \Z$.
    Let $k, n \in \N$,$p_1, \ldots, p_k \in \{0, (X-1)^2, (X-1)^3\}$, and $f_{ij} \in \Z[X^{\pm}]$, $i = 1, \ldots, n; \; j = 1, \ldots, k,$ be as defined in Proposition~\ref{prop:linequndec}.
    We will construct $k$ words $w_1, w_2, \ldots, w_k$ over $\mX \cup \mX^{-1} \cup G$, that has solution over $\Z \wr \Z$ if and only if the system of equations~\eqref{eq:sysmono} has integer solutions $z_1, \ldots, z_n$.

    For each $i = 1, \ldots, n,$ and $j = 1, \ldots, k$, define $h_{ij} \coloneqq (f_{ij}, 0)$, and $h_{0j} \coloneqq (- f_{0j}, 0)$.
    Let $\mX$ be the alphabet $\{x_0, x_1, \ldots, x_n\} \cup \{x_{11}, x_{12}, x_{13}, x_{14}\} \cup \cdots \cup \{x_{k1}, x_{k2}, x_{k3}, x_{k4}\}$.
    For each $j = 1, \ldots, k,$ define the word $w_j$ over the alphabet $\mX \cup \mX^{-1} \cup G$:
    \begin{align*}
        w_j \coloneqq
        \begin{cases}
            (x_0 h_{0j} x_0^{-1}) (x_1 h_{1j} x_1^{-1}) \cdots (x_n h_{nj} x_n^{-1}) \quad & \text{ if } p_j = 0, \\
            (x_0 h_{0j} x_0^{-1}) (x_1 h_{1j} x_1^{-1}) \cdots (x_n h_{nj} x_n^{-1}) [[x_{j1}, x_{j2}], x_{j3}] \quad & \text{ if } p_j = (X-1)^2, \\
            (x_0 h_{0j} x_0^{-1}) (x_1 h_{1j} x_1^{-1}) \cdots (x_n h_{nj} x_n^{-1}) [[[x_{j1}, x_{j2}], x_{j3}], x_{j4}] \quad & \text{ if } p_j = (X-1)^3.
        \end{cases}
    \end{align*}
    We claim that there exist $g_{j1}, g_{j2}, g_{j3}, g_{j4} \in \Z \wr \Z$, such that $g_0 = (F_0, b_0), g_1 = (F_1, b_1), \ldots, g_n = (F_n, b_n)$ satisfy $w_j(g_0, g_1, \ldots) = e$, if and only if $p_j \mid X^{b_1} f_{1j} + \cdots + X^{b_n} f_{nj} - X^{b_0} f_{0j}$.

    Indeed, as in the proof of Theorem~\ref{thm:quadeq}, we have
    \[
     (x_0 h_{0j} x_0^{-1}) (x_1 h_{1j} x_1^{-1}) \cdots (x_n h_{nj} x_n^{-1}) = \left( - X^{b_{0}} f_0 + X^{b_1} f_{1j} + \cdots + X^{b_n} f_{nj}, 0 \right)
    \]
    If $p_j = 0$, then $w_j(g_0, g_1, \ldots) = e$ if and only if $- X^{b_{0}} f_0 + X^{b_1} f_{1j} + \cdots + X^{b_n} f_{nj} = 0$.
    If $p_j = (X - 1)^2$, then by Lemma~\ref{lem:comms}(2), $w_j(g_0, g_1, \ldots) = e$ if and only if $p_j \mid - X^{b_{0}} f_0 + X^{b_1} f_{1j} + \cdots + X^{b_n} f_{nj}$.
    If $p_j = (X - 1)^3$ then Lemma~\ref{lem:comms}(3) gives the same result.

    Therefore, the system of $k$ word equations $w_1 , w_2, \ldots, w_k$, have solutions $g_0, g_1, \ldots, g_n$, $g_{11}, g_{12}$, $g_{13}, g_{14}$, $\ldots$, $g_{k1}, g_{k2}, g_{k3}, g_{k4} \in \Z \wr \Z$, if and only if the system of equations
    \begin{align}\label{eq:sysmono3}
        p_1 & \mid X^{b_1} f_{11} + \cdots + X^{b_n} f_{n1} - X^{b_0} f_{01}, \nonumber\\
        p_2 & \mid X^{b_1} f_{12} + \cdots + X^{b_n} f_{n2} - X^{b_0} f_{02}, \nonumber\\
        & \vdots \nonumber\\
        p_k & \mid X^{b_1} f_{1k} + \cdots + X^{b_n} f_{nk} - X^{b_0} f_{0k}.
    \end{align}
    admit solutions $b_0, b_1, \ldots, b_n \in \Z$.
    Note that this is equivalent to the system~\eqref{eq:sysmono} admitting solutions $z_1, \ldots, z_n \in \Z$, by de-homogenizing or homogenizing the equations.
    Therefore, Proposition~\ref{prop:linequndec} shows that it is undecidable whether a system of equations have solutions in $\Z \wr \Z$.
\end{proof}

\section{Knapsack Problem}\label{sec:KP}
In this section, we prove Theorem~\ref{thm:knapsack}, showing that the Knapsack Problem is undecidable in the abelian-by-cyclic group $\mA \rtimes \Z$, where $\mA$ is the $\Z[X^{\pm}]$-module constructed in Theorem~\ref{thm:linequndec}.
Undecidability of the Knapsack Problem in certain abelian-by-cyclic groups has already been observed by Mishchenko and Treier~\cite[Theorem~1]{mishchenko2017knapsack}.
Nevertheless, we deduce an alternative and more direct proof from Theorem~\ref{thm:linequndec}.
\thmknapsack*

\begin{proof}
    Take the integer $n \in \N$ and the $\Z[X^{\pm}]$-module $\mA$ from Theorem~\ref{thm:linequndec}.
    First, we show that in the abelian-by-cyclic group $\mA \rtimes \Z$, there exist $h_1, \ldots, h_n \in \mA \rtimes \Z$, such that the following ``integer'' variant of the Knapsack Problem is undecidable:

    \smallskip
    \noindent \textbf{Input:} an element $h_0 \in \mA \rtimes \Z$.
    
    \noindent \textbf{Question:} whether there exist $(b_1, \ldots, b_{n+1}) \in \Z^{n+1}$ that satisfy the following equation:
    \begin{equation}\label{eq:weakKP}
        (\bzer, b_1) \cdot h_1 \cdot (\bzer, b_2) \cdot h_2 \cdots (\bzer, b_n) \cdot h_n \cdot (\bzer, b_{n+1}) = h_0.
    \end{equation}
    \smallskip

    We will reduce the problem from Theorem~\ref{thm:linequndec}.
    Let $\mA$ be the finitely presented $\Z[X^{\pm}]$-module constructed in Theorem~\ref{thm:linequndec}.
    Take $h_1 = (\bff_1, 0), \ldots, h_n = (\bff_n, 0)$ where $\bff_1, \ldots, \bff_n \in \mA$ are defined in Theorem~\ref{thm:linequndec}.
    For an instance of the decision problem in Theorem~\ref{thm:linequndec} with input $\bff_0 \in \mA$,
    we take the input $h_0 = (\bff_0, 0)$ in the above decision problem.
    We show that Equation~\eqref{eq:weakKP} has a solution $(b_1, \ldots, b_{n+1}) \in \Z^{n+1}$ if and only if 
    \begin{equation}\label{eq:modeq}
        X^{z_1} \bff_1 + \cdots + X^{z_n} \bff_n = \bff_0.
    \end{equation}
    has a solution $(z_1, \ldots, z_n) \in \Z^n$.
    
    Indeed, direct computation shows that Equation~\eqref{eq:weakKP} is equivalent to 
    \begin{equation}\label{eq:weakKPequiv}
    \left(X^{b_1} \bff_1 + X^{b_1 + b_2} \bff_2 + \cdots + X^{b_1 + b_2 + \cdots + b_n} \bff_n, b_1 + \cdots + b_n + b_{n+1} \right) = \left(\bff_0, 0 \right).
    \end{equation}
    If this has a solution $(b_1, \ldots, b_{n+1}) \in \Z^{n+1}$, then $z_1 \coloneqq b_1, z_2 \coloneqq b_1 + b_2, \ldots, z_n \coloneqq b_1 + \cdots + b_n$, is a solution for \eqref{eq:modeq}.
    Conversely, if $(z_1, \ldots, z_n) \in \Z^n$ is a solution for \eqref{eq:modeq}, then $b_1 \coloneqq z_1, b_2 \coloneqq z_2 - z_1, \ldots, b_n \coloneqq z_n - z_{n-1}, b_{n+1} \coloneqq -z_n$, is a solution for \eqref{eq:weakKPequiv}.
    By Theorem~\ref{thm:linequndec}, it is undecidable whether Equation~\eqref{eq:modeq} has integer solutions. Therefore, it is also whether Equation~\eqref{eq:weakKP} (equivalently, \eqref{eq:weakKPequiv}) has solutions $(b_1, \ldots, b_{n+1}) \in \Z^{n+1}$.

    We then prove Theorem~\ref{thm:knapsack}.
    Equation~\eqref{eq:weakKP} can be rewritten as
    \begin{equation}\label{eq:intKP}
    g_1^{b_1} g_2^{b_2} \cdots g_{n+1}^{b_{n+1}} = g,
    \end{equation}
    where
    \begin{equation*}
    g_1 = (\bzer, 1), \; g_2 = h_1 (\bzer, 1) h_1^{-1}, \; g_3 = h_1 h_2 (\bzer, 1) h_2^{-1} h_1^{-1}, \ldots, \;
    g_{n+1} = h_1 h_2 \cdots h_n (\bzer, 1) h_n^{-1} \cdots h_2^{-1} h_1^{-1},
    \end{equation*}
    and $g = h_0 h_n^{-1} \cdots h_2^{-1} h_1^{-1}$.
    Finally, Equation~\eqref{eq:intKP} has integer solutions $b_1, \ldots, b_{n+1} \in \Z$ if and only if the knapsack equation
    \begin{equation}\label{eq:nKP}
    g_1^{z_1} \left(g_1^{-1}\right)^{z'_1} g_2^{z_2} \left(g_2^{-1}\right)^{z'_2} \cdots g_{n+1}^{z_{n+1}} \left(g_{n+1}^{-1}\right)^{z'_{n+1}} = g
    \end{equation}
    has \emph{non-negative} integer solutions $z_1, z'_1, z_2, z'_2, \ldots, z_{n+1}, z'_{n+1} \in \N$.
    Therefore, it is undecidable for the input $g$, whether Equation~\eqref{eq:nKP} has solutions $z_1, z'_1, z_2, z'_2, \ldots, z_{n+1}, z'_{n+1} \in \N$.
\end{proof}

\section{Coset Intersection}\label{sec:inter}
In this section we show that Coset Intersection is decidable in all finitely generated abelian-by-cyclic groups:

\thmcoset*

For reference, an overview of our decision procedure for Theorem~\ref{thm:coset} is given in Algorithm~\ref{alg:cosetinter}.
Our solution depends on effective computation in finitely presented modules over polynomial rings.
We will make use of the following classic results.

\begin{lemma}[{\cite[Lemma~2.1,~2.2]{baumslag1981computable}}]\label{lem:classicdec}
    Let $\mA$ be a $\Z[X^{\pm}]$-module with a given finite presentation.
    The following problems are effectively solvable:
    \begin{enumerate}[nosep, label = (\roman*)]
        \item \textit{(Submodule Membership)} Given elements $\ba_1, \ldots, \ba_k, \ba \in \mA$, decide whether $\ba$ is in the submodule generated by $\ba_1, \ldots, \ba_k$.
        \item \textit{(Computing Syzygies)} Given elements $\ba_1, \ldots, \ba_k \in \mA$, compute a finite set of generators for the \emph{Syzygy module} $S \subseteq \Z[X^{\pm}]^k$:
        \[
        S \coloneqq \left\{(f_1, \ldots, f_k) \in \Z[X^{\pm}]^k \;\middle|\; f_1 \cdot \ba_1 + \cdots + f_k \cdot \ba_k = \bzer \right\}.
        \]
        \setcounter{DecideCounter}{\value{enumi}} 
    \end{enumerate}
\end{lemma}
In particular, for Lemma~\ref{lem:classicdec}(ii), it states that one can compute the generators for the solution set of any homogeneous linear equation.

The following lemma shows we can effectively compute the intersection of a submodule of $\Z[X^{\pm}]^k$ with $\Z^k$.

\begin{lemma}[{\cite[Corollary~2.5(2)]{baumslag1981computable}}]\label{lem:decinterZ}
    Suppose we are given $k \in \N$ and elements $\bg_1, \ldots, \bg_n$ of the $\Z[X^{\pm}]$-module $\Z[X^{\pm}]^k$.
    Let $\mM$ denote the $\Z[X^{\pm}]$-module generated by $\bg_1, \ldots, \bg_n$, and define $\Lambda \coloneqq \mM \cap \Z^k$.
    Then $\Lambda \subseteq \Z^k$ is a $\Z$-module, and a finite set of generators for $\Lambda$ can be effectively computed.
\end{lemma}

Let $d \geq 1$ be a positive integer. One can define the Laurent polynomial ring over the variable $X^d$:
\[
\Z[X^{\pm d}] \coloneqq \left\{\sum_{i = p}^q a_{di} X^{di} \in \Z[X^{\pm}] \;\middle|\; p, q \in \Z, a_{dp}, \ldots, a_{dq} \in \Z \right\}.
\]
The elements of $\Z[X^{\pm d}]$ are Laurent polynomials in $\Z[X^{\pm}]$ whose monomials have degrees divisible by $d$.
For any $d \geq 1$, the ring $\Z[X^{\pm}]$ is naturally a $\Z[X^{\pm d}]$-module generated by the elements $1, X, \ldots, X^{d-1}$, and $\Z[X^{\pm}]$ is isomorphic as a $\Z[X^{\pm d}]$-module to $\Z[X^{\pm d}]^{d}$.
Any finitely presented $\Z[X^{\pm}]$-module can be considered as a finitely presented $\Z[X^{\pm d}]$-module:

\begin{restatable}{lemma}{lemchangebase}\label{lem:changebase}
    Let $d \geq 2$.
    Given a finite presentation of a $\Z[X^{\pm}]$-module $\mA$, one can compute a finite presentation of $\mA$ as a $\Z[X^{\pm d}]$-module.
    Furthermore, let $\ba \in \mA$ be given in the finite presentation of $\mA$ as $\Z[X^{\pm}]$-module, then one can compute the representation of $\ba$ in $\mA$ considered as a $\Z[X^{\pm d}]$-module.
\end{restatable}

Let $\mA$ be a $\Z[X^{\pm}]$-module given by a finite presentation.
Recall that we naturally identify $\mA$ with the subgroup $\left\{ (\ba, 0) \;\middle|\; \ba \in \mA\right\}$ of $\mA \rtimes \Z$; that is, we will sometimes write $\ba$ instead of $(\ba, 0)$ when the context is clear.
The following lemma effectively describes finitely generated subgroups of $\mA \rtimes \Z$.
Such a description follows from the general description of subgroups of finitely generated metabelian groups~\cite[proof of Theorem~1]{romanovskii1974some}.
Here we give a systematic reformulation in the context of abelian-by-cyclic groups.
\begin{restatable}[Structural theorem of abelian-by-cyclic groups, see also~\cite{romanovskii1974some}]{lemma}{lemstruct}\label{lem:struct}
    Let $\gen{\mG}$ be a subgroup of $\mA \rtimes \Z$ generated by the elements $\mG = \{g_1 \coloneqq (\ba_1, z_1), \ldots, g_K \coloneqq (\ba_K, z_K)\}$.
    Then
    \begin{enumerate}[nosep, label = (\roman*)]
        \item If $z_1 = \cdots = z_K = 0$, then $\gen{\mG}$ is contained in $\mA$ and it is the $\Z$-module generated by $\ba_1, \ldots, \ba_K$.
        \item If $z_1, \ldots, z_K$ are not all zero, then $\gen{\mG} \not\subset \mA$.
        Let $d \in \N$ denote the greatest common divisor of $z_1, \ldots, z_K$. Consider the lattice \[
        \Lambda \coloneqq \left\{(s_1, \ldots, s_K) \in \Z^K \;\middle|\; s_1 z_1 + \cdots + s_K z_K = 0\right\}.
        \]
        Let $(s_{11}, \ldots, s_{1K}), \ldots, (s_{T1}, \ldots, s_{TK})$ be a finite set of generators for $\Lambda$.
        Then $\gen{\mG} \cap \mA$ is a $\Z[X^{\pm d}]$-submodule of $\mA$, generated by the set of elements
        \begin{equation}\label{eq:genmod}
        S \coloneqq \left\{g_i g_j g_i^{-1} g_j^{-1} \;\middle|\; 1 \leq i < j \leq K \right\} \cup \left\{g_1^{s_{i1}} \cdots g_K^{s_{iK}} \;\middle|\; i \in [1, T] \right\}.
        \end{equation}
        \item In case (ii), let $\ba \in \mA$ be any element such that $(\ba, d) \in \gen{\mG}$.
        Then $\gen{\mG}$ is generated by $\gen{\mG} \cap \mA$ and $(\ba, d)$ as a group.
        In other words, every element of $\gen{\mG}$ can be written as $(\bb, 0) \cdot (\ba, d)^m$ for some $\bb \in \gen{\mG} \cap \mA$ and $m \in \Z$.
    \end{enumerate}
\end{restatable}

We point out that in case~(ii), the subgroup $\gen{\mG} \cap \mA$ is finitely generated as a $\Z[X^{\pm d}]$-module; but it is not necessarily finitely generated as a group.

\begin{example}
    Let $\mA = \Z[X^{\pm}]$, considered as a $\Z[X^{\pm}]$-module.
    Let $\gen{\mG}$ be the subgroup of $\mA \rtimes \Z$ generated by the elements $g_1 = (X, 4), g_2 = (1+X, -6)$.
    Then $d = 2$, and $\gen{\mG} \cap \mA$ is the $\Z[X^{\pm 2}]$-module generated by the elements
    \begin{equation*}
    g_1 g_2 g_1^{-1} g_2^{-1} = (X^5 + X^4 - 1 - X^{-5}, 0), \quad
    g_1^3 g_2^2 = (X^{13} + X^{12} + X^{9} + X^{7} + X^{6} + X^{5} + X, 0).
    \end{equation*}

    For example, consider the element $g_1^2 g_2 g_1 g_2 \in \gen{\mG}$.
    By direct computation, its second entry is zero, therefore $g_1^2 g_2 g_1 g_2 \in \gen{\mG} \cap \mA$.
    Furthermore, $g_1^2 g_2 g_1 g_2$ can be written as 
    \begin{multline*}
        g_1^2 g_2 g_1 g_2 = g_1^2 (g_2 g_1) g_2 = g_1^2 (g_2 g_1 g_2^{-1} g_1^{-1}) (g_1 g_2) g_2 = g_1^2 (g_1 g_2 g_1^{-1} g_2^{-1})^{-1} g_1 g_2^2 \\
        = g_1^2 (g_1 g_2^2) (g_1 g_2^2)^{-1} (g_1 g_2 g_1^{-1} g_2^{-1})^{-1} (g_1 g_2^2) = g_1^3 g_2^2 \cdot (g_1 g_2^2)^{-1} (g_1 g_2 g_1^{-1} g_2^{-1})^{-1} (g_1 g_2^2) \\
        = (X^{13} + X^{12} + X^{9} + X^{7} + X^{6} + X^{5} + X, 0) \cdot (g_1 g_2^2)^{-1} (X^5 + X^4 - 1 - X^{-5}, 0)^{-1} (g_1 g_2^2) \\
        = (X^{13} + X^{12} + X^{9} + X^{7} + X^{6} + X^{5} + X) + X^{-8} \cdot (-1) \cdot (X^5 + X^4 - 1 - X^{-5}).
    \end{multline*}
    It is therefore indeed in the $\Z[X^{\pm 2}]$-module generated by $g_1 g_2 g_1^{-1} g_2^{-1} = X^5 + X^4 - 1 - X^{-5}$ and $g_1^3 g_2^2 = X^{13} + X^{12} + X^{9} + X^{7} + X^{6} + X^{5} + X$.

    Intuitively, modulo the generator $g_1 g_2 g_1^{-1} g_2^{-1}$, one can permute letters in any word over $\mG$ (in the above example, $g_1^2 g_2 g_1 g_2$ is congruent to $g_1^3 g_2^2$).
    Whereas the generator $g_1^3 g_2^2$ guarantees the second entry of the product to be zero.
\end{example}

Given finite sets $\mG, \mH \subseteq G$ and $h \in G$, Coset Intersection asks to decide whether $\gen{\mG} \cap h \gen{\mH} = \emptyset$.
We split into three cases according to whether $\gen{\mG}$ and $\gen{\mH}$ are contained in the subgroup $\mA$.
If at least one of $\gen{\mG}$ and $\gen{\mH}$ is contained in $\mA$ (Case 1 and 2 below), then the solutions to Subgroup and Coset Intersection are relatively straightforward using the standard algorithms in Lemma~\ref{lem:classicdec} and \ref{lem:decinterZ}.
If neither $\gen{\mG}$ nor $\gen{\mH}$ is contained in the subgroup $\mA$ (Case 3 below), then the solution is more complicated and we need to invoke Theorem~\ref{thm:onemono}.
The detailed procedure is summarized in Algorithm~\ref{alg:cosetinter}.

\subsection*{Case 1: $\gen{\mG}$ and $\gen{\mH}$ are both contained in $\mA$}

Suppose $\gen{\mG}$ is generated by the elements $g_1 = (\ba_1, 0), \ldots, g_K = (\ba_K, 0)$, and $\gen{\mH}$ is generated by the elements $h_1 = (\ba'_1, 0), \ldots, h_M = (\ba'_M, 0)$.

In this case, we have
\[
\gen{\mG} = \{y_1 \cdot \ba_1 + \cdots + y_K \cdot \ba_K \mid y_1, \ldots, y_K \in \Z\}, \quad \gen{\mH} = \{z_1 \cdot \ba'_1 + \cdots + z_M \cdot \ba'_M \mid z_1, \ldots, z_M \in \Z\}.
\]

Let $h = (\ba_h, z_h)$.
If $z_h \neq 0$ then $\gen{\mG} \cap h \gen{\mH} = \emptyset$.
Therefore we only need to consider the case where $z_h = 0$.
Then $\gen{\mG} \cap h \gen{\mH} = \emptyset$ if and only if there is no solution for $y_1 \cdot \ba_1 + \cdots + y_K \cdot \ba_K = z_1 \cdot \ba'_1 + \cdots + z_M \cdot \ba'_M + z \cdot \ba_h, \; y_1, \ldots, y_K, z_1, \ldots, z_M \in \Z, z = 1$.

Let $\mM'$ denote the $\Z[X^{\pm}]$-module
\begin{multline}\label{eq:M1c}
\mM' \coloneqq \Big\{ (y_1, \ldots, y_K, z_1, \ldots, z_M, z) \in \Z[X^{\pm}]^{K+M+1} \;\Big|\; \\
y_1 \cdot \ba_1 + \cdots + y_K \cdot \ba_K - z_1 \cdot \ba'_1 - \cdots - z_M \cdot \ba'_M - z \cdot \ba_h = \bzer \Big\}.
\end{multline}

\begin{observation}
    We have $\gen{\mG} \cap h \gen{\mH} = \emptyset$ if and only if $\big(\mM' \cap \Z^{K+M+1}\big) \cap \big(\Z^{K+M} \times \{1\}\big) = \emptyset$.
\end{observation}

The generators of $\mM'$ as a $\Z[X^{\pm}]$-module can be computed by Lemma~\ref{lem:classicdec}(ii).
Then, a $\Z$-basis for $\mM' \cap \Z^{K+M+1}$ can be computed by Lemma~\ref{lem:decinterZ}. Therefore, whether $\big(\mM' \cap \Z^{K+M+1}\big) \cap \big(\Z^{K+M} \times \{1\}\big) = \emptyset$ can be decided using linear algebra over $\Z$.

\subsection*{Case 2: one of $\gen{\mG}$ and $\gen{\mH}$ is contained in $\mA$}

We can without loss of generality suppose $\gen{\mH} \subseteq \mA$ and $\gen{\mG} \not\subset \mA$.
Otherwise notice that $\gen{\mG} \cap h \gen{\mH} = \emptyset$ if and only if $h^{-1} \gen{\mG} \cap \gen{\mH} = \emptyset$, so we can exchange the role of $\gen{\mG}$ and $\gen{\mH}$.

By Lemma~\ref{lem:struct}, suppose $\gen{\mG}$ is generated by the element $(\ba_{\mG}, d_{\mG})$ and the $\Z[X^{\pm d_{\mG}}]$-module $\gen{\mG} \cap \mA$, and $\gen{\mH}$ is generated by the elements $h_1 = (\ba'_1, 0), \ldots, h_M = (\ba'_M, 0)$.
The generators of the $\Z[X^{\pm d_{\mG}}]$-module $\gen{\mG} \cap \mA$ can be effectively computed by Lemma~\ref{lem:struct}.
Also, by Lemma~\ref{lem:changebase}, we can consider $\mA$ as a finitely presented $\Z[X^{\pm d_{\mG}}]$-module instead of a $\Z[X^{\pm}]$-module, and suppose the generators of $\gen{\mG} \cap \mA$ as well as $\ba'_1, \ldots, \ba'_M$ are given as elements of the $\Z[X^{\pm d_{\mG}}]$-module $\mA$.

Let $h = (\ba_h, z_h)$.
If $d_{\mG} \nmid z_h$ then $\gen{\mG} \cap h \gen{\mH} = \emptyset$.
Therefore we only need to consider the case where $z_h = z d_{\mG}$ for some $z \in \Z$.
Then $\gen{\mG} \cap h \gen{\mH} \neq \emptyset$ if and only if the equation
$
(\bb, 0) \cdot (\ba_{\mG}, d_{\mG})^z = (\ba_h, z_h) \cdot (\bc, 0)
$
has solutions $\bb \in \gen{\mG} \cap \mA,\; \bc \in \sum_{i = 1}^M \Z \cdot \ba'_i$.
Direct computation shows this is equivalent to
\[
    X^{z_h} \cdot \bc + \left(\ba_h - \frac{X^{z d_{\mG}} - 1}{X^{d_{\mG}} - 1} \cdot \ba_{\mG} \right) = \bb.
\]

Let $\mM'$ denote the $\Z[X^{\pm d_{\mG}}]$-module
\begin{multline}\label{eq:M2c}
\mM' \coloneqq \Bigg\{ (z_1, \ldots, z_M, z) \in \Z[X^{\pm d_{\mG}}]^{M+1} \;\Bigg|\; \\
X^{z_h} \cdot \left(z_1 \cdot \ba'_1 + \cdots + z_M \cdot \ba'_M\right) + z \cdot \left(\ba_h - \frac{X^{z_h} - 1}{X^{d_{\mG}} - 1} \cdot \ba_{\mG} \right) \in \gen{\mG} \cap \mA \Bigg\}. 
\end{multline}

\begin{observation}
We have $\gen{\mG} \cap h \gen{\mH} = \emptyset$ if and only if $d_{\mG} \nmid z_h$ and $\big(\mM' \cap \Z^{M+1}\big) \cap \big(\Z^{M} \times \{1\}\big) = \emptyset$.
\end{observation}

The generators of $\mM'$ as a $\Z[X^{\pm d_{\mG}}]$-module can be computed by Lemma~\ref{lem:classicdec}(ii), applied over the quotient module $\mA / \left(\gen{\mG} \cap \mA \right)$.
Then, a $\Z$-basis for $\mM' \cap \Z^{M+1}$ can be computed by Lemma~\ref{lem:decinterZ}. 
Therefore, whether $\big(\mM' \cap \Z^{M+1}\big) \cap \big(\Z^{M} \times \{1\}\big) = \emptyset$ can be decided using linear algebra over $\Z$.

\subsection*{Case 3: neither $\gen{\mG}$ nor $\gen{\mH}$ is contained in $\mA$}

By Lemma~\ref{lem:struct}, suppose $\gen{\mG}$ is generated by $\gen{\mG} \cap \mA$ and an element $(\ba_{\mG}, d_{\mG})$; and suppose $\gen{\mH}$ is generated by $\gen{\mH} \cap \mA$ and an element $(\ba_{\mH}, d_{\mH})$.
The elements $(\ba_{\mG}, d_{\mG})$ and $(\ba_{\mH}, d_{\mH})$ can be effectively computed from the generating sets $\mG, \mH$ by performing the Euclidean algorithm.
Furthermore, $\gen{\mG} \cap \mA$ is a $\Z[X^{\pm d_{\mG}}]$-module whose generators are explicitly given (by Equation~\eqref{eq:genmod}), and $\gen{\mH} \cap \mA$ is a $\Z[X^{\pm d_{\mH}}]$-module whose generators are explicitly given.

Let $h = (\ba_h, z_h)$.
Then $\gen{\mG} \cap h \gen{\mH} = \emptyset$ if and only if the equation
$
    (\bb, 0) \cdot (\ba_{\mG}, d_{\mG})^m = (\ba_h, z_h) \cdot (\bc, 0) \cdot (\ba_{\mH}, d_{\mH})^n
$
has solutions $\bb \in \gen{\mG} \cap \mA, \bc \in \gen{\mH} \cap \mA$, and $m, n \in \Z$.
By direct computation, this is equivalent to the system
\begin{equation}\label{eq:Cosetsys}
    \bb + \frac{X^{m d_{\mG}} - 1}{X^{d_{\mG}} - 1} \cdot \ba_{\mG} = \ba_h + X^{z_h} \cdot \bc + X^{z_h} \cdot \frac{X^{n d_{\mH}} - 1}{X^{d_{\mH}} - 1} \cdot \ba_{\mH}, \quad m d_{\mG} = n d_{\mH} + z_h.
\end{equation}

We define $d \coloneqq \lcm(d_{\mG}, d_{\mH})$ to be the least common multiplier of $d_{\mG}$ and $d_{\mH}$, and consider both $\gen{\mG} \cap \mA$ and $\gen{\mH} \cap \mA$ as $\Z[X^{\pm d}]$-modules, respectively generated by the sets $S_{\mG}$ and $S_{\mH}$.
The sets $S_{\mG}$ and $S_{\mH}$ can be effectively computed by Lemma~\ref{lem:changebase} from the generators of $\gen{\mG} \cap \mA$ as a $\Z[X^{\pm d_{\mG}}]$-module, and from the generators of $\gen{\mH} \cap \mA$ is a $\Z[X^{\pm d_{\mH}}]$-module.
We define the $\Z[X^{\pm d}]$-module 
\[
    \mM' \coloneqq (\gen{\mG} \cap \mA) + X^{z_h} \cdot (\gen{\mH} \cap \mA),
\]
which is generated by the set $S_{\mG} \cup \{X^{z_h} \cdot s \mid s \in S_{\mH}\}$.

If $m d_{\mG} = n d_{\mH} + z_h$ has no integer solutions $m, n$, then $\gen{\mG} \cap h \gen{\mH} = \emptyset$.
Otherwise, there exist $z_{\mG} \coloneqq m d_{\mG}, z_{\mH} \coloneqq n d_{\mH} \in \Z$ such that $d_{\mG} \mid z_{\mG},\; d_{\mH} \mid z_{\mH}$ and $z_{\mG} = z_{\mH} + z_h$.
Then, every solution $(m, n) \in \Z^2$ of the equation $m d_{\mG} = n d_{\mH} + z_h$ is of the form
\[
m = (z_{\mG} + z d)/d_{\mG}, \quad n = (z_{\mH} + z d)/d_{\mH}, \quad z \in \Z.
\]

\begin{restatable}{lemma}{lemcosettoeq}\label{lem:cosettoeq}
    Let $z_{\mG}, z_{\mH}$ be integers such that $d_{\mG} \mid z_{\mG},\; d_{\mH} \mid z_{\mH}$ and $z_{\mG} = z_{\mH} + z_h$.
    The intersection $\gen{\mG} \cap h \gen{\mH}$ is non-empty if and only if the equation
        \begin{equation}\label{eq:Coseteqz}
            X^{z d} \cdot \ba'_{\mG, \mH} - \ba''_{\mG, \mH} \in (X^d - 1) \cdot \mM'
        \end{equation}
        has solution $z \in \Z$.
        Here,
        \[
        \ba'_{\mG, \mH} \coloneqq X^{z_{\mG}} \cdot \left(\frac{X^d - 1}{X^{d_{\mG}} - 1} \cdot \ba_{\mG} - \frac{X^d - 1}{X^{d_{\mH}} - 1} \cdot \ba_{\mH}\right),
        \]
        \[
        \ba''_{\mG, \mH} \coloneqq \frac{X^d - 1}{X^{d_{\mG}} - 1} \cdot \ba_{\mG} - X^{z_h} \cdot \frac{X^d - 1}{X^{d_{\mH}} - 1} \cdot \ba_{\mH} + (X^d - 1) \cdot \ba_h.
        \]
\end{restatable}
\begin{proof}
    Suppose the $\gen{\mG} \cap h \gen{\mH}$ is non-empty. Let $(\ba, z') \in \gen{\mG} \cap h \gen{\mH}$, then $d_{\mG} \mid z'$ and $d_{\mH} \mid (z' - z_h)$.
    Hence $z' = z_{\mG} + zd = z_{\mH} + zd + z_h$ for some $z \in \Z$.
    Since $(\ba, z') \in \gen{\mG} \cap h \gen{\mH}$, the Equation~\eqref{eq:Cosetsys} has solution with $m d_{\mG} = z' = z_{\mG} + zd, \; n d_{\mH} = z' - z_h = z_{\mH} + zd$, meaning
    \begin{multline*}
        \frac{X^{z d} \cdot \ba'_{\mG, \mH} - \ba''_{\mG, \mH}}{X^d - 1} =
        \frac{X^{z_{\mG} + zd}}{X^{d_{\mG}} - 1} \cdot \ba_{\mG} - \frac{X^{z_{\mG} + zd}}{X^{d_{\mH}} - 1} \cdot \ba_{\mH} - \frac{1}{X^{d_{\mG}} - 1} \cdot \ba_{\mG} + \frac{X^{z_h}}{X^{d_{\mH}} - 1} \cdot \ba_{\mH} - \ba_h \\
        = \frac{X^{z_{\mG} + zd} - 1}{X^{d_{\mG}} - 1} \cdot \ba_{\mG} - \frac{X^{z_{\mG} + zd} - X^{z_h}}{X^{d_{\mH}} - 1} \cdot \ba_{\mH} - \ba_h \\
        = \frac{X^{z_{\mG} + zd} - 1}{X^{d_{\mG}} - 1} \cdot \ba_{\mG} - X^{z_h} \cdot \frac{X^{z_{\mH} + zd} - 1}{X^{d_{\mH}} - 1} \cdot \ba_{\mH} - \ba_h = X^{z_h} \cdot \bc - \bb \in \mM'.
    \end{multline*}
    Therefore \eqref{eq:Coseteqz} is satisfied.

    Conversely, suppose Equation~\eqref{eq:Coseteqz} is satisfied.
    Then we have
    \[
    \frac{X^{z_{\mG} + zd} - 1}{X^{d_{\mG}} - 1} \cdot \ba_{\mG} - X^{z_h} \cdot \frac{X^{z_{\mH} + zd} - 1}{X^{d_{\mH}} - 1} \cdot \ba_{\mH} - \ba_h = \frac{X^{z d} \cdot \ba'_{\mG, \mH} - \ba''_{\mG, \mH}}{X^d - 1} \in \mM',
    \]
    so it can be written as $X^{z_h} \cdot \bc - \bb$ for some $\bb \in \gen{\mG} \cap \mA, \; \bc \in \gen{\mH} \cap \mA$.
    Hence the system~\eqref{eq:Cosetsys} has solutions $\bb \in \gen{\mG} \cap \mA,\; \bc \in \gen{\mH} \cap \mA,\; m = \frac{z_{\mG} + z d}{d_{\mG}},\; n = \frac{z_{\mH} + z d}{d_{\mH}}$.
\end{proof}


By Lemma~\ref{lem:cosettoeq}, it suffices to decide whether Equation~\eqref{eq:Coseteqz} has a solution $z \in \Z$.
To do this, consider both $\mA$ and $(X^d - 1) \cdot \mM'$ as $\Z[X^{\pm d}]$-modules.
Equation~\eqref{eq:Coseteqz} is equivalent to the equation
\begin{equation}\label{eq:inquot}
\left(X^{d}\right)^z \cdot \ba'_{\mG, \mH} = \ba''_{\mG, \mH}
\end{equation}
in the quotient module $\mA / \left((X^d - 1) \cdot \mM'\right)$.
Applying Theorem~\ref{thm:onemono} to the variable $X^d$ (instead of $X$) and the finitely presented $\Z[X^{\pm d}]$-module $\mA / \left((X^d - 1) \cdot \mM'\right)$, we can decide whether the above equation has a solution $z \in \Z$.

Combining the three cases, we have proven Theorem~\ref{thm:coset}.

\begin{algorithm}[!h]
\caption{Algorithm for Coset Intersection}
\label{alg:cosetinter}
\begin{description} 
\item[Input:]
a finite presentation of the $\Z[X^{\pm}]$-module $\mA$, two finite sets of elements $\mG = \{(\ba_1, z_1), \ldots, (\ba_K, z_K)\},\; \mH = \{(\ba'_1, z'_1), \ldots, (\ba'_M, z'_M)\}$ in the group $\mA \rtimes \Z$, an element $h = (\ba_h, z_h)$.
\item[Output:] \textbf{True} (when $\gen{\mG} \cap h \gen{\mH} = \emptyset$) or \textbf{False} (when $\gen{\mG} \cap h \gen{\mH} \neq \emptyset$).
\end{description}
\begin{enumerate}[nosep, label=\arabic*.]
    \item \textbf{If $z_1, \ldots, z_K$ and $z'_1, \ldots, z'_M$ are all zero.}
    
    Compute generators of the module $\mM'$ defined in Equation~\eqref{eq:M1c}.
    
    Decide whether 
    $
    \big(\mM' \cap \Z^{K+M+1}\big) \cap \big(\Z^{K+M} \times \{1\}\big) = \emptyset
    $
    using Lemma~\ref{lem:classicdec} and \ref{lem:decinterZ}.
    If yes, return \textbf{True}, otherwise return \textbf{False}.
    
    \item \textbf{If one of the sets $\{z_1, \ldots, z_K\}$ and $\{z'_1, \ldots, z'_M\}$ is all zero.}
    
    Without loss of generality suppose $z'_1 = \cdots = z'_M = 0$, otherwise swap the sets $\mG, \mH$ and replace $h$ with $h^{-1}$.
    \begin{enumerate}[nosep, label = (\roman*)]
        \item Compute $d_{\mG} \coloneqq \gcd(z_1, \ldots, z_K)$, and compute generators of the $\Z[X^{\pm d_{\mG}}]$-module $\gen{\mG} \cap \mA$ using Lemma~\ref{lem:struct}.
        \item If $d_{\mG} \mid z_h$, continue, otherwise return \textbf{True}.
        \item Compute generators of the modules $\mM'$ defined in Equation~\eqref{eq:M2c}.
        
        Decide whether 
        $
        \big(\mM' \cap \Z^{M+1}\big) \cap \big(\Z^{M} \times \{1\}\big) = \emptyset
        $
        using Lemma~\ref{lem:classicdec} and \ref{lem:decinterZ}.
        If yes, return \textbf{True}, otherwise return \textbf{False}.
    \end{enumerate}
    \item \textbf{If none of the sets $\{z_1, \ldots, z_K\}$ and $\{z'_1, \ldots, z'_M\}$ is all zero.}
    \begin{enumerate}[nosep, label = (\roman*)]
        \item Compute $d_{\mG} \coloneqq \gcd(z_1, \ldots, z_K), d_{\mH} \coloneqq \gcd(z'_1, \ldots, z'_M), d \coloneqq \gcd(d_{\mG}, d_{\mH})$.
        \item Decide whether the equation $m d_{\mG} = n d_{\mH} + z_h$ has solutions $(m, n) \in \Z^2$.
        If yes, take any solution $(m, n)$ and let $z_{\mG} \coloneqq m d_{\mG},\; z_{\mH} \coloneqq n d_{\mH}$; otherwise return \textbf{True}.
        \item Compute the generators of the $\Z[X^{\pm d_{\mG}}]$-module $\gen{\mG} \cap \mA$ and the generators of the $\Z[X^{\pm d_{\mH}}]$-module $\gen{\mH} \cap \mA$ using Lemma~\ref{lem:struct}.
        Compute their respective generators $S_{\mG}, S_{\mH}$ as $\Z[X^{\pm d}]$-modules using Lemma~\ref{lem:changebase}.
        Let $\mM'$ be the $\Z[X^{\pm d}]$-module generated by $S_{\mG} \cup \{X^{z_h} \cdot s \mid s \in S_{\mH}\}$.
        \item 
        Decide whether the Equation~\eqref{eq:inquot} over $\mA/\left((X^d - 1) \cdot \mM'\right)$ has a solution $z \in \Z$, using Theorem~\ref{thm:onemono}.
        If yes, return \textbf{False}, otherwise return \textbf{True}.
    \end{enumerate}
\end{enumerate}
\end{algorithm}

\newpage

\bibliography{intermeta}

\begin{thebibliography}{GKLZ18}

\bibitem[BCM23]{DBLP:conf/icalp/Benedikt0M23}
Michael Benedikt, Dmitry Chistikov, and Alessio Mansutti.
\newblock The complexity of {P}resburger arithmetic with power or powers.
\newblock In {\em 50th International Colloquium on Automata, Languages, and
  Programming, {ICALP} 2023}, volume 261 of {\em LIPIcs}, pages 112:1--112:18.
  Schloss Dagstuhl - Leibniz-Zentrum f{\"{u}}r Informatik, 2023.

\bibitem[BCMI81]{baumslag1981computable}
Gilbert Baumslag, Frank~B. Cannonito, and Charles~F. Miller~III.
\newblock Computable algebra and group embeddings.
\newblock {\em Journal of Algebra}, 69(1):186--212, 1981.

\bibitem[BCR94]{baumslag1994algorithmic}
Gilbert Baumslag, Frank~B. Cannonito, and Derek~J.S. Robinson.
\newblock The algorithmic theory of finitely generated metabelian groups.
\newblock {\em Transactions of the American Mathematical Society},
  344(2):629--648, 1994.

\bibitem[Bod24]{bodart2024membership}
Corentin Bodart.
\newblock Membership problems in nilpotent groups.
\newblock {\em arXiv preprint arXiv:2401.15504}, 2024.

\bibitem[Bol76]{boler1976conjugacy}
James Boler.
\newblock Conjugacy in abelian-by-cyclic groups.
\newblock {\em Proceedings of the American Mathematical Society}, 55(1):17--21,
  1976.

\bibitem[CCZ20]{DBLP:conf/icalp/CadilhacCZ20}
Micha{\"{e}}l Cadilhac, Dmitry Chistikov, and Georg Zetzsche.
\newblock Rational subsets of {B}aumslag-{S}olitar groups.
\newblock In {\em 47th International Colloquium on Automata, Languages, and
  Programming, {ICALP}}, volume 168 of {\em LIPIcs}, pages 116:1--116:16.
  Schloss Dagstuhl - Leibniz-Zentrum f{\"{u}}r Informatik, 2020.

\bibitem[CE19]{DBLP:conf/icalp/CiobanuE19}
Laura Ciobanu and Murray Elder.
\newblock Solutions sets to systems of equations in hyperbolic groups are
  {EDT0L} in {PSPACE}.
\newblock In {\em 46th International Colloquium on Automata, Languages, and
  Programming, {ICALP} 2019}, volume 132 of {\em LIPIcs}, pages 110:1--110:15.
  Schloss Dagstuhl - Leibniz-Zentrum f{\"{u}}r Informatik, 2019.

\bibitem[Cul81]{culler1981using}
Marc Culler.
\newblock Using surfaces to solve equations in free groups.
\newblock {\em Topology}, 20(2):133--145, 1981.

\bibitem[DE17]{DBLP:conf/icalp/DiekertE17}
Volker Diekert and Murray Elder.
\newblock Solutions of twisted word equations, {EDT0L} languages, and
  context-free groups.
\newblock In {\em 44th International Colloquium on Automata, Languages, and
  Programming, {ICALP} 2017}, volume~80 of {\em LIPIcs}, pages 96:1--96:14.
  Schloss Dagstuhl - Leibniz-Zentrum f{\"{u}}r Informatik, 2017.

\bibitem[DVZ18]{delgado2018intersection}
Jordi Delgado, Enric Ventura, and Alexander Zakharov.
\newblock Intersection problem for {D}roms {RAAG}s.
\newblock {\em International Journal of Algebra and Computation},
  28(07):1129--1162, 2018.

\bibitem[Eis13]{eisenbud2013commutative}
David Eisenbud.
\newblock {\em Commutative algebra: with a view toward algebraic geometry},
  volume 150.
\newblock Springer Science \& Business Media, 2013.

\bibitem[FGLZ20]{DBLP:conf/icalp/FigeliusGLZ20}
Michael Figelius, Moses Ganardi, Markus Lohrey, and Georg Zetzsche.
\newblock The complexity of knapsack problems in wreath products.
\newblock In {\em 47th International Colloquium on Automata, Languages, and
  Programming, {ICALP} 2020}, volume 168 of {\em LIPIcs}, pages 126:1--126:18.
  Schloss Dagstuhl - Leibniz-Zentrum f{\"{u}}r Informatik, 2020.

\bibitem[FM00]{farb2000asymptotic}
Benson Farb and Lee Mosher.
\newblock On the asymptotic geometry of abelian-by-cyclic groups.
\newblock {\em Acta Mathematica}, 184(2):145--202, 2000.

\bibitem[GKLZ18]{DBLP:conf/stacs/GanardiKLZ18}
Moses Ganardi, Daniel K{\"{o}}nig, Markus Lohrey, and Georg Zetzsche.
\newblock Knapsack problems for wreath products.
\newblock In {\em 35th Symposium on Theoretical Aspects of Computer Science,
  {STACS} 2018}, volume~96 of {\em LIPIcs}, pages 32:1--32:13. Schloss Dagstuhl
  - Leibniz-Zentrum f{\"{u}}r Informatik, 2018.

\bibitem[GMO20]{garreta2020diophantine}
Albert Garreta, Alexei Miasnikov, and Denis Ovchinnikov.
\newblock Diophantine problems in solvable groups.
\newblock {\em Bulletin of Mathematical Sciences}, 10(01):2050005, 2020.

\bibitem[HX21]{hurtado2021global}
Sebastian Hurtado and Jinxin Xue.
\newblock Global rigidity of some abelian-by-cyclic group actions on {T}2.
\newblock {\em Geometry \& Topology}, 25(6):3133--3178, 2021.

\bibitem[KLM20]{kharlampovich2020diophantine}
Olga Kharlampovich, Laura L{\'o}pez, and Alexei Myasnikov.
\newblock The {D}iophantine {p}roblem in some metabelian groups.
\newblock {\em Mathematics of Computation}, 89(325):2507--2519, 2020.
\newblock Corrected version: arxiv.org/abs/1903.10068.

\bibitem[Lev22]{levine2022equations}
Alex Levine.
\newblock Equations in virtually class 2 nilpotent groups.
\newblock {\em journal of Groups, Complexity, Cryptology}, 14, 2022.

\bibitem[LSZ15]{lohrey2015rational}
Markus Lohrey, Benjamin Steinberg, and Georg Zetzsche.
\newblock Rational subsets and submonoids of wreath products.
\newblock {\em Information and Computation}, 243:191--204, 2015.

\bibitem[LU21]{lysenok2021orientable}
Igor Lysenok and Alexander Ushakov.
\newblock Orientable quadratic equations in free metabelian groups.
\newblock {\em Journal of Algebra}, 581:303--326, 2021.

\bibitem[Luk82]{luks1982isomorphism}
Eugene~M. Luks.
\newblock Isomorphism of graphs of bounded valence can be tested in polynomial
  time.
\newblock {\em Journal of computer and system sciences}, 25(1):42--65, 1982.

\bibitem[LZ20]{lohrey2020knapsack}
Markus Lohrey and Georg Zetzsche.
\newblock Knapsack and the power word problem in solvable {B}aumslag-{S}olitar
  groups.
\newblock In {\em 45th International Symposium on Mathematical Foundations of
  Computer Science (MFCS 2020)}. Schloss Dagstuhl-Leibniz-Zentrum f{\"u}r
  Informatik, 2020.

\bibitem[Mak77]{makanin1977problem}
Gennadiy~Semenovich Makanin.
\newblock The problem of solvability of equations in a free semigroup.
\newblock {\em Matematicheskii Sbornik}, 145(2):147--236, 1977.

\bibitem[Mak83]{makanin1983equations}
Gennady~S. Makanin.
\newblock Equations in a free group.
\newblock {\em Mathematics of the USSR-Izvestiya}, 21(3):483, 1983.

\bibitem[Mat93]{matiyasevich1993hilbert}
Yuri~V. Matiyasevich.
\newblock {\em Hilbert's 10th Problem}.
\newblock MIT Press, Cambridge, Massachusetts, 1993.

\bibitem[MNU15]{myasnikov2015knapsack}
Alexei Myasnikov, Andrey Nikolaev, and Alexander Ushakov.
\newblock Knapsack problems in groups.
\newblock {\em Mathematics of Computation}, 84(292):987--1016, 2015.

\bibitem[MT17]{mishchenko2017knapsack}
Alexei Mishchenko and Alexander Treier.
\newblock Knapsack problem for nilpotent groups.
\newblock {\em Groups Complexity Cryptology}, 9(1):87--98, 2017.

\bibitem[MU23]{mandel2023quadratic}
Richard Mandel and Alexander Ushakov.
\newblock Quadratic equations in metabelian {B}aumslag-{S}olitar groups.
\newblock {\em arXiv preprint arXiv:2302.06974}, 2023.

\bibitem[Nar96]{narendran1996solving}
Paliath Narendran.
\newblock Solving linear equations over polynomial semirings.
\newblock In {\em Proceedings 11th Annual IEEE Symposium on Logic in Computer
  Science}, pages 466--472. IEEE, 1996.

\bibitem[Nos82]{noskov1982conjugacy}
Gennady~Andreevich Noskov.
\newblock Conjugacy problem in metabelian groups.
\newblock {\em Mathematical notes of the Academy of Sciences of the USSR},
  31:252--258, 1982.

\bibitem[PS19]{potapov2019vector}
Igor Potapov and Pavel Semukhin.
\newblock Vector and scalar reachability problems in
  $\operatorname{SL}_2(\mathbb{Z})$.
\newblock {\em Journal of Computer and System Sciences}, 100:30--43, 2019.

\bibitem[PSC03]{pittet2003random}
Christophe Pittet and Laurent Saloff-Coste.
\newblock Random walks on abelian-by-cyclic groups.
\newblock {\em Proceedings of the American Mathematical Society},
  131(4):1071--1079, 2003.

\bibitem[Rom74]{romanovskii1974some}
N.~S. Romanovskii.
\newblock Some algorithmic problems for solvable groups.
\newblock {\em Algebra and Logic}, 13(1):13--16, 1974.

\bibitem[Rom12]{roman2012equations}
Vitalii Roman'kov.
\newblock Equations over groups.
\newblock {\em Groups - Complexity - Cryptology}, 4(2):191--239, 2012.

\bibitem[Sch80a]{schreyer1980berechnung}
Frank-Olaf Schreyer.
\newblock Die {B}erechnung von {S}yzygien mit dem verallgemeinerten
  {W}eierstra{\ss}schen {D}ivisionssatz.
\newblock Master's thesis, Fakult{\"a}t f{\"u}r Mathematik, Universit{\"a}t
  Hamburg, 1980.

\bibitem[Sch80b]{schupp1980quadratic}
Paul~E. Schupp.
\newblock Quadratic equations in groups, cancellation diagrams on compact
  surfaces, and automorphisms of surface groups.
\newblock In {\em Studies in Logic and the Foundations of Mathematics},
  volume~95, pages 347--371. Elsevier, 1980.

\bibitem[Sem80]{semenov1980certain}
Aleksei~L. Semenov.
\newblock On certain extensions of the arithmetic of addition of natural
  numbers.
\newblock {\em Mathematics of the USSR-Izvestiya}, 15(2):401, 1980.

\bibitem[UW23]{ushakov2023quadratic}
Alexander Ushakov and Chloe Weiers.
\newblock Quadratic equations in the lamplighter group.
\newblock {\em arXiv preprint arXiv:2401.08589}, 2023.

\end{thebibliography}

\appendix
\section{Omitted proofs}\label{app:proof}

\thmonemono*
\begin{proof}
We give a deduction of Theorem~\ref{thm:onemono} from Noskov's Lemma:
\begin{lemma}[{Noskov's Lemma~\cite{noskov1982conjugacy}, see also~\cite[Proposition~2.4]{baumslag1994algorithmic}}]
    There is an algorithm which, given a finitely generated commutative ring $R$ and a finite subset $S$ of the group of units $U(R)$, finds a finite presentation of the multiplicative subgroup $\langle S \rangle$.
\end{lemma}

Consider the ideal
\begin{equation}\label{eq:defI}
\mI \coloneqq \left\{f \in \Z[X^{\pm}] \;\middle|\; f \cdot \bff_1 = \bzer \right\}
\end{equation}
of $\Z[X^{\pm}]$, then a finite set of generators for $\mI$ can be computed by Lemma~\ref{lem:classicdec}(ii).

Using Lemma~\ref{lem:classicdec}(i), decide whether there exists $h_0 \in \Z[X^{\pm}]$ such that $h_0 \cdot \bff_1 = \bff_0$.
If such $h_0$ does not exist, then $X^z \cdot \bff_1 = \bff_0$ has no solution $z \in \Z$.
Otherwise, if such an $h_0$ exists, it can be found by enumerating through $\Z[X^{\pm}]$.
The solution set
\[
    \left\{f \in \Z[X^{\pm d}] \;\middle|\; f \cdot \bff_1 = \bff_0 \right\}
\]
is equal to $h_0 + \mI \coloneqq \{h_0 + f \mid f \in \mI\}$.
Consider the ideal $\mJ$ of $\Z[X^{\pm}, Y^{\pm}]$ generated by $\mI$ and the element $Y - h_0$.
Then $X^z \in h_0 + \mI$ if and only if $X^z - Y \in \mJ$, which is equivalent to the equation $X^z = Y$ in the quotient ring $\Z[X^{\pm}, Y^{\pm}]/\mJ$.
We then apply Noskov's Lemma on the quotient ring $R = \Z[X^{\pm}, Y^{\pm}]/\mJ$ and the subset $S = \{X, Y\}$ of its group of units.
This determines whether there exists $z \in \Z$ such that $X^z = Y$ in the quotient ring $\Z[X^{\pm}, Y^{\pm}]/\mJ$.
\end{proof}

\lemchangebase*
\begin{proof}
    Write $\mA = M/N$ where $M, N \subseteq \Z[X^{\pm}]^D$.
    Every element $\bg \in \Z[X^{\pm}]^D$ can be uniquely written as $\bg = \bg_0 + X \cdot \bg_1 + \cdots + X^{d-1} \cdot \bg_{d-1}$ where $\bg_0, \bg_1, \ldots, \bg_{d-1}$ are in $\Z[X^{\pm d}]^D$.
    This gives an effective isomorphism $\varphi \colon \Z[X^{\pm}]^D \rightarrow \Z[X^{\pm d}]^{Dd}$.
    The generators of $\varphi(M)$ can be obtained by simply applying $\varphi$ to the generators of $M$, similarly for $\varphi(N)$.
    Hence $\mA = \varphi(M)/\varphi(N)$ is a finite presentation of $\mA$ as a $\Z[X^{\pm d}]$-module.
\end{proof}

\lemstruct*
\begin{proof}
    (i) is obvious.
    For (ii), suppose $z_1, \ldots, z_K$ are not all zero and let $d \in \N$ be their greatest common divisor.
    Let $n_1, \ldots, n_K \in \Z$ be such that $n_1 z_1 + \cdots + n_K z_K = d$, then $g \coloneqq g_1^{n_1} \cdots g_K^{n_K}$ is of the form $(\bb, d), \bb \in \mA$.
    Then for any $\ba \in \gen{\mG} \cap \mA$, we have $\gen{\mG} \cap \mA \ni g^{-1} \ba g = (- X^{-d} \cdot \bb, -d) (\ba, 0) (\bb, d) = (X^d \cdot \ba, 0) = X^d \cdot \ba$.
    Similarly, $\gen{\mG} \cap \mA \ni g \ba g^{-1} = X^{-d} \cdot \ba$.
    Therefore, $\gen{\mG} \cap \mA$ is a $\Z[X^{\pm d}]$-module.

    On one hand, the elements in $S$ are obviously in $\gen{\mG} \cap \mA$, so the $\Z[X^{\pm d}]$-module generated by $S$ is a submodule of $\gen{\mG} \cap \mA$.
    On the other hand, we show that the quotient $(\gen{\mG} \cap \mA)/S$ is trivial.
    Notice that since $(g_i g_j g_i^{-1} g_j^{-1})^{-1} = g_j g_i g_j^{-1} g_i^{-1}$, we have $\left\{g_i g_j g_i^{-1} g_j^{-1} \;\middle|\; i, j \in [1, K] \right\} \subseteq S$.
    Therefore, the quotient by $S$ allows one to permute elements in any product $g_{i_1}^{\epsilon_1} \cdots g_{i_n}^{\epsilon_n} \in \gen{\mG} \cap \mA$ without changing their class in $(\gen{\mG} \cap \mA)/S$.
    More precisely, for $g, g' \in \gen{\mG}$ and $i, j \in [1, K]$, if $g g_i g_j g' \in \mA$ then we have $g g_i g_j g' + S = g g_j g_i g' + S$.
    Indeed, we have $g g_i g_j g' = g (g_i g_j g_i^{-1} g_j^{-1}) g_j g_i g' = g (g_i g_j g_i^{-1} g_j^{-1}) g^{-1} + g g_j g_i g'$ and $g (g_i g_j g_i^{-1} g_j^{-1}) g^{-1}$ is in the $\Z[X^{\pm d}]$-module generated by $S$.
    For every product $g_{i_1}^{\epsilon_1} \cdots g_{i_n}^{\epsilon_n} \in \gen{\mG} \cap \mA$, we must have 
    \[
    \left(\sum_{j \in [1, n], i_j = 1} \epsilon_j, \ldots, \sum_{j \in [1, n], i_j = K} \epsilon_j\right) \in \Lambda
    \]
    by looking at the second component.
    Since $(s_{11}, \ldots, s_{1K}), \ldots, (s_{T1}, \ldots, s_{TK})$ be are the generators for $\Lambda$, by permuting the elements in the product we can rewrite $g_{i_1}^{\epsilon_1} \cdots g_{i_n}^{\epsilon_n}$ as 
    \[
    \left(g_1^{s_{1 1}} \cdots g_K^{s_{1 K}}\right)^{j_1} \cdots \left(g_1^{s_{T 1}} \cdots g_K^{s_{T K}}\right)^{j_T},
    \]
    where $j_1, \ldots, j_T \in \Z$.
    Therefore, $g_{i_1}^{\epsilon_1} \cdots g_{i_n}^{\epsilon_n}$ is in the $\Z[X^{\pm d}]$-module generated by $S$.

    For (iii), $\ba \in \mA$ be any element such that $(\ba, d) \in \gen{\mG}$.
    Since $d \in \N$ is the greatest common divisor for $z_1, \ldots, z_K$, every element $g$ of $\gen{\mG}$ must be of the form $(\bc, md), \; \bc \in \mA, m \in \Z$.
    Then $g \cdot (\ba, d)^{-m} \in \gen{\mG} \cap \mA$. Let $(\bb, 0) \coloneqq g \cdot (\ba, d)^{-m}$, then $\ba \in \gen{\mG} \cap \mA$ and $g = (\bb, 0) \cdot (\ba, d)^m$.
\end{proof}

\end{document}